\documentclass{ecai}

\newcommand{\card}[1]{\lvert #1 \rvert}

\usepackage{graphicx} 

\usepackage{amsmath,amsfonts,amscd}
\usepackage{latexsym}
\usepackage{amssymb}
\usepackage{booktabs}
\usepackage{enumitem}
\usepackage{esvect}
\usepackage{color}
\usepackage{amsthm}
\usepackage{algorithm}
\usepackage{algpseudocode}
\usepackage{tikz}
\usepackage{subcaption}
\usepackage{float} 

\newtheorem{observation}{Observation}

\newtheorem{example}{Example} 
\newtheorem{lemma}{Lemma} 
\newtheorem{theorem}{Theorem} 
\newtheorem{proposition}{Proposition} 
\newtheorem{corollary}{Corollary} 

\theoremstyle{definition}
\newtheorem{definition}{Definition}[section]

\newcommand{\BibTeX}{B\kern-.05em{\sc i\kern-.025em b}\kern-.08em\TeX}

\newcommand{\bd}{\partial}
\newcommand{\fami}{\mathbb{F}}
\newcommand{\absl}[1]{\lVert #1 \rVert}

\begin{document}

\begin{frontmatter}


\paperid{5755} 


\title{Optimal Candidate Positioning in Multi-Issue Elections}


\author[A]{\fnms{Colin}~\snm{Cleveland}\orcid{0000-0003-0909-7518}\thanks{Corresponding Author. Email: colin.cleveland@kcl.ac.uk}\footnote{Equal contribution.}}
\author[A]{\fnms{Bart}~\snm{de Keijzer}\orcid{0000-0001-9465-0837}\footnotemark}
\author[A]{\fnms{Maria}~\snm{Polukarov}\orcid{0000-0002-7421-3012}} 

\address[A]{King's College London}



\begin{abstract}
We study strategic \emph{candidate positioning} in multidimensional
spatial‑voting elections.  Voters and candidates are represented as
points in~$\mathbb{R}^{d}$ and each voter supports the candidate that is closest under a distance induced by an $\ell_{p}$‑norm.  We prove that computing an optimal location for a new candidate is \mbox{NP‑hard} already against a \emph{single} opponent, whereas for a constant number of issues the problem is tractable: an $O(n^{d+1})$ hyperplane‑enumeration algorithm and an $O(n\log n)$ radial‑sweep routine for~$d=2$ solve the task exactly.  We further derive the first approximation guarantees for the general multi‑candidate case and show how our geometric approach extends seamlessly to positional‑scoring rules such as $k$‑approval and
Borda.  These results clarify the algorithmic landscape of multidimensional spatial elections and provide practically implementable tools for campaign strategy.
\end{abstract}

\end{frontmatter}


\maketitle              

\section{Introduction}


In an ideal democracy, elections are a dynamic exchange between voters and candidates, not merely a ritual for selecting leaders. Voters convey their policy preferences, and candidates adjust their positions in response. Determining an optimal stance is challenging in a competitive, multi-issue space. We address this using a spatial voting framework to model voters and candidates, aiming to analyse computational complexity and design algorithms for optimal candidate positioning.

We use a natural spatial voting framework of \citet{spatialvoting} that models voters' opinions and candidates' policies within the $d$-dimensional space $\mathbb{R}^d$, where each dimension represents a quantifiable issue. For example, one coordinate might represent income tax rate, with a voter's position reflecting their ideal rate. If there are $d$ relevant issues, each is represented by a coordinate in $\mathbb{R}^d$. Each voter then prefers the candidate whose policy position is closest to their ideal point, according to a given distance measure (such as the Euclidean distance).

Indeed, when voters evaluate candidates, it is natural to assume that they prefer the candidate whose proposed policies are ``closer'' to their own ideology. While the distance measure used does not need to adhere to a specific form, distance measures induced by unweighted and weighted $\ell^p$ norms with $p \geq 1$ are commonly employed in computational social choice research (e.g., \citet{lu_manipulating_2019,estornell_election_2020,wu_manipulating_2022,imber_spatial_2024}), due to both their mathematically appealing properties and their intuitive interpretation. 



In general voting scenarios, a voter's favourite candidate may not always get their vote due to strategic voting (e.g., \citet{myatt_theory_2007}). However, in this work, we assume voters are non-strategic and vote truthfully for their favourite candidate. Instead, we focus on candidates' strategic behaviour to optimize their positions to win. Unlike most studies where candidates' positions are fixed (e.g., \citet{lu_manipulating_2019,estornell_election_2020}), we explore a setting where candidates can change their positions, similar to the perception-changing manipulations by ~\citet{wu_manipulating_2022}.


\subsection*{Our Contributions}

Our contributions are as follows:

\begin{itemize}
    \item \textbf{Single-opponent setting (Section~\ref{sec:onecandidate}).}  
    We provide several $\mathsf{NP}$-completeness results, as well as algorithms that solve the optimal positioning problem in polynomial time when $d$ (the dimension of the issue space) is fixed. Our main algorithm runs in time $O(n^{d+1})$, and we additionally present an $O(n \log n)$ algorithm for the case $d=2$. These algorithms are compatible with $\ell^p$-norm distance measures, with $p > 1, p \in \mathbb{N}$,\footnote{The condition that $p \in \mathbb{N}$ is imposed in order to be compatible with computation over the rational numbers: we assume that the coordinates of the candidates and voters in the input are rational, so that the $\ell^p$ distance between two points is always rational when $p \in \mathbb{N}$. Our results straightforwardly carry over to the case where $p \not\in \mathbb{N}$, although it would require appropriately relaxing the standard model of computation to allow for computation with arbitrary real or algebraic numbers.} and are based on a translation of the problem into enumerating the set of \emph{connected} regions of a central hyperplane arrangement \cite{stanley2007hyperplane} in $\mathbb{R}^d$.

    \item \textbf{Geometric perspective.}  
    We show that solving the candidate positioning problem is equivalent to finding a point in $\mathbb{R}^d$ contained in the interior of a maximum number of given $d$-dimensional balls, specified by their centre and radius. Algorithms for determining whether such a point exists, under the Euclidean norm, have been proposed by \citet{lushchakova_geometric_2020}, \citet{crama_complexity_1995}, and \citet{cramaibaraki}. While correct in their original contexts, these are not directly applicable to our voting setting, which introduces additional challenges. In Section~\ref{sec:manycandidates}, we provide more details and resolve these issues by adapting Lushchakova’s algorithm.


    \item \textbf{Approximation algorithms.}  
    We make two observations:  
    (1) against a single opponent, there is a straightforward $2$-approximation algorithm, which only needs to evaluate two locations in the issue space; and  
    (2) our algorithm for the single-candidate case yields a polynomial-time $m$-approximation algorithm for the multi-candidate case, where $m$ is the number of candidates.


    \item \textbf{Extension to other voting systems.}  
    Throughout the main sections, we assume plurality voting, i.e., each voter casts a single vote for the closest candidate. In Appendix~\ref{apx:generalisations}, we describe a generalisation to positional voting systems (where voters assign scores to multiple candidates), and we show that our model and results extend naturally to this setting.


\end{itemize}



\section{Related work}\label{sec:relwork}
\emph{Spatial voting} 
refers to a body of research focused on electoral models that capture voters' perceptions of candidates on various political issues and their corresponding voting behaviours. Specifically, much of the work in this direction has explored and analysed the relationships between voting behaviour and perceptions of political issues.

In theoretical studies of political systems, it is often assumed that voters will choose candidates based on their proposed policies~\cite{powell_elections_2000}. However, empirical political studies have also examined real-life case studies that consider candidates' personality traits and their impact on election outcomes, such as in~\citet{avery_state-level_2015}. This line of research, particularly highlighted in the paper by \citet{falcao_we_2023}, challenges the idea that candidates should be represented solely by their stances on political issues. Importantly though, note that aspects such as perceived personality traits can be naturally incorporated as additional dimensions in a multi-dimensional issue space, thus justifying the approach of modelling a candidate as a point within such space. Moreover, \citet{bogomolnaia_euclidean_2007} point out, in particular, that when the dimension of the issue space, $d$, is sufficiently large, spatial voting can encompass every complete preference profile.


Spatial voting in political economics has been introduced by  ~\citet{570930af-c28a-3bc3-be3a-5052d0478611}, who formulated a spatial model in an economics context, and which was later translated by \citet{downs} into voting theory, yielding the \emph{Hotelling-Downs} framework. Further analysis of numerous election models based on the Hotelling-Downs model have since been carried out, e.g.,~\citet{feldman2016variations,brusco2012hotelling}. Arrow's work \cite{arrow_social_1951} proved the existence of optimal position choices for candidates embedded in a $1$-dimensional Euclidean space, as implied by the Median Voter Theorem.

The study of the higher-dimensional space $\mathbb{R}^d$, which forms the basis in our present study, dates back to the 1976 work of \citet{spatialvoting}, where the authors focused on characterising certain equilibrium points such as plurality and Condorcet points. In the computer science community, the problem of algorithmically constructing such points has been considered by \citet{wuetal}, \citet{linetal}, and \citet{bergplurality}. 
We also refer to the book of Enelow et al.~\cite{enelow_spatial_1984} for a unified treatment of classical work on spatial voting in strategic settings with more than one dimension. 
 


Specifically,  \citet{lu_manipulating_2019,estornell_election_2020} investigated a type of manipulation called \emph{issue selection}, wherein the manipulator can influence voters' preferences by either selecting only a subset of issues or adjusting the weights assigned to different issues. The concept of \emph{perception changing} was explored by Wu~\cite{wu_manipulating_2022}, who considered moving the candidates' positions within a certain range to potentially alter the election outcome. Another avenue in spatial voting research involves scenarios with \emph{incomplete information}: \citet{imber_spatial_2024} analyse the challenge of determining a potential winner when voters' positions are not fully known.

In the single-opponent setting, several prior works are indirectly related. From the developments in Section~\ref{sec:onecandidate} and the observations of \citet{wuetal}, this case is connected to computing a point of maximum \emph{Tukey depth}, and algorithms such as \citet{chan_mechanism_2021,bremnertukey,liutukey} can find a location that maximises the number of voters preferring the candidate to the opponent (i.e., the size of the candidate’s win set). These approaches, however, are tied to the Euclidean norm and do not apply to general distance measures. Our results also imply a translation to the \emph{densest hemisphere problem} studied by \citet{densesthemisphere}, which likewise assumes Euclidean distance and does not output a candidate location in $\mathbb{R}^d$. Finally, \citet{deliberative} propose an algorithm that, when adapted to our setting, yields the largest win set for a fixed opponent, but again only for the Euclidean norm and without producing the corresponding position. \footnote{The algorithm proposed in \cite{deliberative} is, just like ours, based on hyperplane arrangements. We want to point out that our approach is nonetheless fundamentally different: Their approach appeals to an algorithm following from a proof in \cite{kearnsvazirani}, which uses hyperplane arrangements in general position, rather than central hyperplane arrangements used in our method.} The algorithms in the above works all make use of invariance properties of the Euclidean norm and do not seem to generalise to other distance measures. The algorithms we propose in the present work can be used in conjunction with a wide spectrum of distance measures induced by $\ell^p$-norms, and are able to output a concrete optimal location for the new candidate (where, in addition, we take into account bit complexity; i.e., the coordinates output by the algorithm require a polynomially bounded number of bits). 

In most algorithmic studies, manipulation problems become $\mathsf{NP}$-hard once the number of issues is unbounded. A notable real-world contrast is provided by Lichtman’s ``Keys to the White House’’ model~\cite{lichtman_predicting_2020}, which suggests that only a small set of pivotal issues typically determines election outcomes. This is consistent with the idea that voters have limited attention for political issues, making a low-dimensional issue space a plausible practical assumption. In such cases, one can model the election by setting $d$ to the number of pivotal issues and projecting the data onto those dimensions, discarding the rest.


Beyond the context of spatial voting, our setting is also related to the work of \citet{crama_complexity_1995} where the problem was framed in a commercial product placement context, with the goal of optimising one's market share with respect to a set of competing products. The complexity of finding a maximum-cardinality, that is, a point contained in as many balls as possible, under both $\ell^2$ and $\ell^{\infty}$ norms has been proven to be $\mathsf{NP}$-complete~\cite{crama_complexity_1995}. In the present work, we establish the $\mathsf{NP}$-completeness of a special case of this problem, where there is only one candidate (or product), using a simpler proof. Furthermore, we strengthen it to hold for $\ell^p$-norms with any $p$ satisfying $1 < p < \infty$. Another context where ball-arrangement problems similar to ours show up is in coalition formation in deliberation spaces, studied in \citet{deliberativefirst} and \citet{deliberative}, where in the latter work an $\mathsf{NP}$-completeness result similar to the one in \cite{crama_complexity_1995} is established. Lastly, the related densest hemisphere problem has been shown to be $\mathsf{NP}$-complete in the work of \citet{densesthemisphere}. Our  $\mathsf{NP}$-completeness proofs drastically simplify all the above $\mathsf{NP}$-completeness results, and generalise those results to arbitrary $\ell^p$-norms.



As mentioned, the studies of \citet{crama_complexity_1995,lushchakova_geometric_2020,cramaibaraki} all offer algorithms for problems closely related to ours: In \citet{crama_complexity_1995} and \citet{cramaibaraki}, algorithms for finding a maximum-cardinality non-empty intersecting family of closed balls is proposed, for the case of the $\ell^2$-norm. Subsequently, Lushchakova presents an algorithm for finding a point in the intersection of a given set of balls, along with a proof of correctness and runtime analysis; in Section~\ref{sec:manycandidates}, we further discuss these algorithms in the context of our setting.

Lastly we mention Voronoi games, introduced by Ahn et al. \cite{voronoigames}, as a domain in computational geometry to which our work relates. These are competitive facility location problems in which two players iteratively place their facilities in a given space, with the objective of maximising their market share. Our results for the case of a single competitor, can be viewed as studying the optimal response of the second player in a one-round discrete Voronoi game (see  \citet{oneroundvoronoi}) where each player has a single facility.


\section{Preliminaries}

\paragraph{High‑level intuition.}  Imagine each voter draws a translucent \emph{influence ball} centred at her ideal point whose radius stretches up to the nearest existing candidate.  A newcomer wins exactly those voters whose balls all overlap at her chosen position.  Our optimisation task is therefore to find a point contained in the interior of as many of these balls as possible.  The formal definitions below make this picture precise.

We use $[a]$ to denote the set of natural numbers $\{1,\ldots,a\}$. For a family of sets $S$, we write $\cap S$ and $\cup S$ to refer to the intersection of all sets in $S$ and the union of all sets in $S$, respectively. 

The objects we consider in this study are pairs of point sets, each of a finite size, in $\mathbb{R}^d$: This models an election scenario where the point sets correspond to positions of candidates and voters in an issue space. A voter is assumed to vote for the candidate who is located closest to them with respect to the distance measure induced by a given norm (in particular, one may think of the standard Euclidean distance). We formally define an \emph{election instance} as follows.

\begin{definition}
    An \emph{election instance} is a tuple $(d,n,m,x,t)$, where $d,n,m \in \mathbb{N}$, $x$ is a list of $n$ points in $\mathbb{Q}^d$, and $t$ is a list of $m$ points in $\mathbb{Q}^d$, disjoint from $x$.
    We refer to $d$ as the \emph{dimension of the issue space}, or simply \emph{dimension}. We refer to $n$ and $m$ as the \emph{number of voters} and the \emph{number of candidates} respectively. The elements in the sets $[n]$ and $[m]$ will be referred to as \emph{voters} and \emph{candidates} respectively. For $i \in [n]$ ($j \in [m]$) the \emph{location of voter $i$ (candidate $j$)} is the point $x_i$ ($t_j$).
\end{definition}

As a notational convention, we will use $i$ to denote a voter and $j$ to denote a candidate. 
We will furthermore work with various norms on $\mathbb{R}^d$, which we will denote by either $\rVert\cdot\lVert$ or $p$.
    



We investigate the question of choosing an optimal location in the issue space from the point of view of a candidate who is purely interested in achieving the most favourable outcome of an election and is not constrained in any way in choosing their position in the issue space. This can be stated as an optimisation problem, for which we aim to design algorithms that solve it and study its computational complexity. What the candidate considers as a ``most favourable'' outcome is context-dependent, but one of the most natural objectives that a candidate could consider is to maximise the total number of votes received (where a voter votes for the candidate closest to them), which is the objective we primarily focus on in this study. Another meaningful objective, which we will also address (albeit less extensively) would be positioning oneself so that the number of votes received is larger than that of as many as possible other candidates. The two objectives are not equivalent, as can be seen from the following example. 

\begin{example}
Consider an election instance $(d,n,m,x,t)$ where $d=1$, $n = 9$, $m = 3$. The locations of the three candidates on the line are $t = (2,20,26)$. The voter locations are $x = (1,1,1,3,3,22,22,24,24)$. Note that, before a newly arriving candidate chooses any position, Candidate 1 receives 5 votes, and Candidates 2 and 3 both receive 2 votes. If a new candidate positions themselves at location 23, then they receive 4 votes, thereby maximising their number of received votes. However, in that case, Candidate 1 will still receive more votes than the new candidate. On the other hand, if the new candidate positions themselves at location 1, the number of votes received by them will be 3, but all other candidates will end up with fewer votes than the new candidate.
\end{example}

We can consider these optimisation problems from either the perspective of a newly arriving candidate (who has to choose a position in $\mathbb{R}^d$ facing a given election instance), or an existing candidate in $[m]$ of the given election instance, who is looking to reposition. It is straightforward to see that these perspectives are equivalent, although, under the latter interpretation, our work can be regarded as a study on computing best responses in an election game where candidates are players who aim to position themselves most favourably. 

Let $I = (d,n,m,x,t)$ be an election instance. Under our model, a voter $i \in [n]$ votes for the candidate $j \in [m]$ whose policy $t_j$ is closest to their location $x_i$. In our study, we choose not to focus on situations where ties need to be broken, so we will assume for the sake of simplicity that a voter $i$ always breaks ties in favour of an existing candidate (rather than the new candidate) in case the new candidate is tied with one or more existing candidates with respect to their distance to $i$. Thus, if the newly arriving candidate wants to win any vote, they should not place themselves precisely on another candidate's location, and if it wants to win the vote of a particular voter $i$, they have to choose their policy at a point strictly closer to $x_i$ than any existing candidate's location. The area in which the new candidate has to position themselves to receive $x_i$'s vote can thus be expressed as the open ball $B_i = \{ t' \in \mathbb{R}^d : \lVert x_i - t' \rVert < \lVert x_i - t_k \rVert, k \in [m] \setminus\{j\} \}$. That is, the open ball centred at $x_i$, with radius $\min\{\lVert x_i - t_k\rVert\ :\ k \in [m]\}$, with $\lVert \cdot \rVert$ denoting a given norm. This gives rise to the following key notions.
\begin{definition}\label{def:keynotions}
For an election instance $I = (d,n,m,x,t)$ and norm $p$, the \emph{critical region} of a voter $i \in [n]$, is the open ball $B_i(I, p) = \{ t' \in \mathbb{R}^d : p(x_i - t') < p(x_i - t_k), k \in [m] \setminus j \}$. 
For a set of voters $S \subseteq [n]$ we write $\mathcal A^K(S,I,p)$ to refer to the intersection area $\cap\{B_i(I,p)\ :\ i \in S\}$. The instance $I$ and/or norm $p$ will usually be clear from context, in which case we simplify notation and write $B_i$ or $B_i(p)$ instead of $B_i(I,p)$. Similarly, we may write 
$\mathcal A^K(S)$ rather than 
$\mathcal A^K(S,I,p)$. 
\end{definition}


To build some further intuition about the structure of an election instance, we may think of it as a  specific arrangement of balls in $\mathbb{R}^d$ and it is useful to briefly consider exactly which such arrangements correspond to election instances. Since no candidate lies inside any of the balls $B_1, \ldots, B_n$, and for every ball there is a candidate on the boundary of it, we arrive at the following observation.

\begin{observation}
Let $\mathcal{B}$ be a set of $n$ balls in $\mathbb{R}^d$, under a given norm $p$. There exists an election instance $I = (d,n,m,x,t)$ such that $\{B_i(I,p)\ :\ i \in [n]\} = \mathcal{B}$ if and only if for all $B \in \mathcal{B}$, there exists a point on the boundary of $B$ that is not in $\cup(\mathcal{B}\setminus \{B\})$.
\end{observation} 




\noindent The problems we study can be formalised as the following: \vspace{1mm} \\
\fbox{\parbox{1.0\columnwidth}{
		\textsc{Optimal Candidate Positioning($U$) (OCP($U$))} \\[1ex]
		\begin{tabular}{ l p{9.6cm} }
			\emph{Input:} & An election instance $I = (d,n,m,x,t)$.
                \\ \vspace{1mm}
            \emph{Parameters:} & A function $U$ that maps an election instance and \\ & point in  $\mathbb{R}^d$ to $\mathbb{N}$. \\ \vspace{1mm}
		\emph{Goal:}  & Find a point $t \in \mathbb{R}^d$ that maximises $U$.
		\end{tabular}
}} 
\vspace{1mm} \\
\noindent In particular, we take for $U$ the vote count function $\nu_p$ as our primary focus, but also will more briefly consider the rank function $R_p$, with respect to a norm $\ell^p$ on $\mathbb{R}^d$, defined as follows. Let $V_p(t) = \{i \in [n] : t \in B_i(p)\}$ be the set of voters who would vote (under norm $p$) for a candidate positioned at $t \in \mathbb{R}^d$.
\begin{align*}
& \nu_p(t)= \card{V(t)} \\
& R_p(t) = \card{\{j \in [m] : }\{i \in [n] \setminus V_p(t): t_j \in Cl(B_i(p))\}\card{ < \nu_p(t)\}}.
\end{align*}
We drop the subscripts~$p$ when clear from context, or when $p$ is irrelevant to the discussion. While our main focus is on the vote-count objective $\nu$, several results also apply to the rank objective $R$, which we define in its tie-aware form (ties in vote counts between the new and existing candidates are resolved by a fixed rule). Both the radial-sweep algorithm for $d=2$ and the hyperplane-region enumeration for $d>2$ enumerate all full-dimensional regions in~$\mathbb{R}^d$ corresponding to distinct voter allocations, ensuring correctness under any tie-handling convention. More generally, $U$ can model other voting rules, such as instant-runoff or arbitrary scoring rules where voters may assign scores even to candidates that are not their closest; see Appendix~\ref{apx:generalisations}.

We also study the decision variant of the above problem: we define OCP-D($U$) as the problem where we are given a pair $(I,k)$, with $I$ an election instance and $k \in \mathbb{N}$, and where we must decide if there exists a point $t' \in \mathbb{R}^d$ such that $U(t') \geq k$.

\section{Competing against a single candidate ($m=1$)}
\label{sec:onecandidate}
We first consider the case $m = 1$, modelling a two-party scenario where a candidate has only a single competitor. In this case, the closures of balls $B_1, \ldots, B_n$ all intersect at $t_1$, and we can assume without loss of generality that $t_1 = \mathbf{0} = (0,\ldots, 0)$.

It turns out working with hyperplanes is useful for analysing the case where $m=1$. We use the following notation: 
\begin{definition}
Let $I = (d,n,m,x,t)$ be an election instance with $m=1$, assume that $t_1 = \mathbf{0}$, and consider a voter $i \in [n]$. 
For a norm $p$, we write $h_{i}(p)$ to denote the hyperplane tangent to $B_i(p)$ at $\mathbf{0}$, and $H_{i}(p)$ to denote the open halfspace with boundary $h_{i}(p)$ containing $B_i(p)$. When clear from context, we drop the reference to the norm $p$ and write simply $h_i$ and $H_i$.   
\end{definition}
Note that when we work with the $\ell^2$ norm, then $x_i$ is actually the normal vector of $h_i$, so the set of points that comprises the hyperplane is given by a linear equation: 
$h_i = \{y : \sum_{j\in[d]} (x_i)_jy_j = 0\}$.


The following lemma states a useful property of the relationship between the critical region and its tangent hyperplanes at $\mathbf{0}$, in case the critical region is convex and has a differentiable boundary (which holds in particular when we use the $\ell^p$-norm with $1 < p < \infty$). For the sake of generality, we state it in terms of arbitrary convex sets with a differentiable boundary.
\begin{lemma}
\label{L:COV-PLN}
Let $B_1, \ldots, B_r \subseteq \mathbb{R}^d$ be a collection of $r$ open $d$-dimensional convex sets, for which there is a unique tangent hyperplane at any point on the boundary, and for which the closures intersect at a point $c \in \mathbb{R}^d$. For $i \in [r]$, let $h_{i,c}$ be the hyperplane at $c$ tangent to $B_i$, and let $H_{i,c}$ be the corresponding open half-space (i.e. such that $h_{i,c}$ forms its respective boundary) containing $B_i$. The intersection $\cap\{H_{i,c} : i\in [r]\}$ is non-empty if and only if the intersection $\cap \{B_i : i \in [r]\}$ is non-empty.
\end{lemma}
\begin{proof}
    If $\cap \{B_i : i \in [r]\} \not = \varnothing$, this intersection is contained in all half-spaces \{$H_{i,c} : i\in [r]\}$, because $B_i$ is  contained in $H_{i,c}$ for all $i \in [r]$, according to our definition. The necessity direction holds.

    For sufficiency: If the $r$ half-spaces intersect, we may pick a point $c'$ in their intersection. Let $i \in [r]$. The line $L_{c,c'}$ passing through $c'$ and $c$ is not contained in $h_{i,c}$. Therefore (by uniqueness of $h_{i,c}$), $L_{c,c'}$ does not lie on a hyperplane that is tangent to $B_i$ at $c$. (Here, we use the fact that every line $L$ tangent to $B_i$ at $c$ must be contained in a hyperplane tangent to $B_i$ at $c$. This is a direct consequence of the hyperplane separation theorem.) Thus, $L_{c,c'}$ intersects with $B_i$ at a point $c_i' \not= c$, and $c_i'$ is in $H_{i,c}$ by definition of $H_{i,c}$ as the half-space containing $B_i$. By convexity of $B_i$, the entire line segment from $c_i'$ to $c$ (excluding $c$) is contained in $B_i$. 
    
    Let $L_{\cap}$ be the intersection of the following $r+1$ line segments: the line segment from $c$ to $c'$, and the line segments from $c$ to $c_i'$, for $i \in [r]$. As we established that $L_{\cap} \setminus c$ is contained in in all of $\{B_i : i \in [r]\}$ the claim follows.
\end{proof}

The set of all hyperplanes $\{H_{i,c} : i \in [n]\}$ forms a \emph{central} hyperplane arrangement in $\mathbb{R}^d$, i.e., all hyperplanes contain the point $\mathbf{0}$. This differs from the usual general-position assumption: in our setting, any subset of at least $d$ hyperplanes meets exactly at the origin. An optimal point must lie in the interior of a full-dimensional region; if it lies on a lower-dimensional face, it can be perturbed into an adjacent full-dimensional region without changing the set of voters gained. See Figure~\ref{fig:flows} for a two-dimensional example illustrating the correspondence between the arrangement of critical-region boundaries and the induced central hyperplane arrangement.

\begin{figure}[H]
    \centering
    \begin{subfigure}[H]{0.23\textwidth}
         \centering
         \resizebox{.8\textwidth}{!}{\begin{tikzpicture}

    \draw[fill=none,line width=5](0,5)  circle (5.0) node [black,yshift=-1.5cm]{};
    \draw[fill=none,line width=5,gray](4,3)  circle (5.0) node [gray,yshift=-1.5cm]{};
    \draw[fill=none,line width=5,yellow](6,-2.5)  circle (6.5) node [gray,yshift=-1.5cm]{};
    \draw[fill=none,line width=5,orange](4,-3)  circle (5) node [gray,yshift=-1.5cm]{};
    \draw[fill=none,line width=5,red](5/3,-12/3)  circle (13/3) node [gray,yshift=-1.5cm]{};

   \filldraw[blue] (0,0) circle (10pt) node[anchor=north west,text = blue]{};

\end{tikzpicture}}
         \caption{The boundaries of the critical regions of the five voters. The opponent candidate is visualised as the blue dot.}
         \label{fig:flow_exmp}
     \end{subfigure}
     \hfill
     \vspace{0.5cm} 
     \begin{subfigure}[H]{0.23\textwidth}
         \centering
         \resizebox{.8\textwidth}{!}{\begin{tikzpicture}
    \draw[line width=5] (-5,0) -- (5,0) {};
    \draw[line width=5,gray] (-3,4) -- (3,-4) {};
    \draw[line width=5,yellow] (2.5,6) -- (-2.5,-6) {};
    \draw[line width=5,orange] (3,4) -- (-3,-4) {};
    \draw[line width=5,red] (4,5/3) -- (-4,-5/3) {};

    \node[below, red, scale=5.0] at (4.5,2) {+};
    \node[below, red, scale=5.0] at (-4.5,-0.2) {-};

    \filldraw[blue] (2,-1) circle (3pt) node[anchor=north west,text = blue]{\Huge A};
    \filldraw[blue] (3,.8) circle (3pt) node[anchor=north west,text = blue]{\Huge B};


\end{tikzpicture}}
         \caption{The corresponding central hyperplane arrangement of all lines tangent to the critical regions at the opponent candidate's position.}
         \label{fig:flow_aux}
    \vspace{0.5cm}
     \end{subfigure}
     
    \caption{Visualisation of an example election instance with five voters and one candidate, together with the induced central hyperplane arrangement.\vspace{0.5cm}}
    \label{fig:flows}
\end{figure}
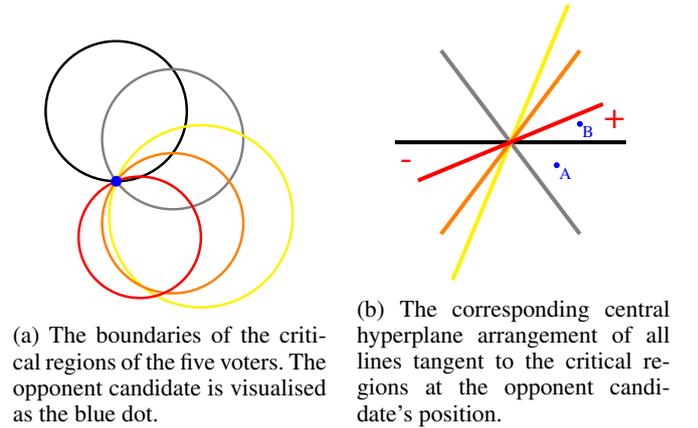


The following lemma characterises the hyperplane tangent to an $\ell^p$-norm ball, and will be useful in the subsequent developments. The proof can be found in Appendix \ref{apx:pballtangent}.

\begin{lemma}\label{lem:pballtangent}
Consider the open $\ell^p$-norm ball $B$ centred at a point $z \in \mathbb{R}^d$ that has $\mathbf{0}$ at its boundary, i.e., $B = \{y\ :\ \lVert y - z \rVert_p = \lVert z \rVert_p\}$. The hyperplane $h$ tangent to $B$ at $\mathbf{0}$ is described by $h = \{y \in \mathbb{R}^d : \sum_{j \in [d]} \text{sign}(z_j)\card{z_j}^{p-1} y_j = 0\}$.
\end{lemma}

\subsection{Computational Hardness}

\textbf{Road‑map.}  We first reduce the homogeneous bipolar variant of \textsc{Max‑Fls} (\textsc{FLS}$^{hb}$) to the single‑opponent case under the Euclidean norm and then observe that the reduction is norm‑independent, which yields hardness for every $\ell_{p}$ with $1<p<\infty$.

\begin{definition}\label{def:fls}
FLS$^{hb}$, standing for ``Feasible Linear Subsystem (homogeneous bipolar version)'', is the computational problem where, given a matrix $A \in \mathbb{Q}^{n \times m}$, and a number $k \in \mathbb{N}$, such that all coefficiencts of $A$ are $\{-1,+1\}$-valued (i.e., $A$'s coordinates are bipolar), we have to decide whether there exists an index set $R \subseteq [n]$ of rows with $\card{R} \geq k$, such that there is a solution $\mathbf{x} \in \mathbb{R}^d$ to the homogeneous linear system $A_R \mathbf{x} > \mathbf{0}$, where $A_R$ is $A$ restricted to the rows indexed by $R$.
\end{definition}

With Lemma\ref{L:COV-PLN} and \ref{lem:pballtangent}, we show that even when $m=1$, the problem OCP-D($\nu_p$) is $\mathsf{NP}$-hard for all choices of $p \in (1,\infty)$. 

We prove this through a polynomial time reduction from a variant of the Maximum Feasible Linear Subsystem problem, which we denote by FLS$^{hb}$, which is an $\mathsf{NP}$-complete problem \cite{amaldi_complexity_1995}. This problem is first reduced to OCP-D($\nu_2$), for the case where $m = 1$. After that, we reduce OCP-D($\nu_2$) for an arbitrary rational-valued choice of $p$. An interesting aspect to the latter reduction is that, despite some analysis, the reduction itself turns out to be trivial and points out that the problems are equivalent in a certain precise sense, and independent of the value of $p$ used in the norm.

\begin{theorem}\label{thm:npcomplete}
OCP-D($\nu_2$) is $\mathsf{NP}$-complete, when restricted to $m=1$, and when all voter coordinates are $\{-1,+1\}$-valued.
\end{theorem}
\begin{proof}
    Containment in $\mathsf{NP}$ follows from Lemma \ref{L:COV-PLN}, as it suffices to simply check for a given point $t \in \mathbb{Q}^d$ whether it lies in at least $k$ of the halfspaces $\{H_i\ :\ i\in[n]\}$.
    
    For hardness, we reduce from FLS$^{hb}$. Let $(A, k)$ be an instance of FLS$^{hb}$ with $A \in \mathbb{Q}^{n' \times m'}$ and $k \in [n']$. We now construct the following instance $(d,n,m,x,t,k)$ from $(A, k)$:
    \begin{enumerate}
    \item Set $d = m'$, $n = n'$, $m = 1$. That is, there is one candidate, the issue space has a dimension equal to the number of columns in $A$, and for each row of $A$ we introduce a voter in the instance we construct.
    \item The location of the single candidate is $t_1 = \mathbf{0}$, (the origin).
    \item For row $i \in [n]$, let the location $x_i$ of the $i$th voter be $a_i \in \mathbb{R}^d$, i.e., the $i$th row of $A$. Note that, as $a_i \in \{-1,+1\}^d$, it holds that $x_i \in \{-1,+1\}^d$.
    \item The $k$ of the constructed OCP-instance is equal to that of the FLS$^{hb}$-instance.
    \end{enumerate}

We claim that in the constructed instance, there is a location $t_2 \in \mathbb{R}^d$ receiving at least $k$ votes if and only if there exists an index set $R \subseteq [n], \card{R} \geq k$ such that there is a solution $\mathbf{y} \in \mathbb{R}^{d}$ to $A_R \mathbf{y} > \mathbf{0}$. 

Let $t_{2}$ be a point that receives at least $k$ votes in the constructed instance. This is equivalent to the existence of a set of voters $R \subseteq [n]$ such that $t_2$ is in the interior of all balls $\{B_i\ :\ i \in R\}$. By Lemma \ref{L:COV-PLN}, this is equivalent to the existence of a point $\mathbf{y} \in \mathbb{R}^d$ in the intersection of the open half spaces $\{H_i : i \in R\}$.

By definition, the coefficients of $h_i$ are given by $x_i = a_i$. Hence, $\mathbf{y}$ satisfies $a_i \cdot \mathbf{y} > 0$ for all $i \in R$, i.e., it is a solution to $A_R \cdot \mathbf{y} > 0$. 
\end{proof}

Next, we show that $\mathsf{NP}$-hardness remains to hold when we work with an $\ell^p$ norm where $p \not= 2$. 

\begin{proposition}\label{prop:hardnesslp}
    OCP-D($\nu_p$) is $\mathsf{NP}$-complete for $1 < p < \infty$, when restricted to $m=1$, and when all voter coordinates are in $\{-1,+1\}$.
\end{proposition}
\begin{proof}
We provide a polynomial-time reduction from the OCP-D($\nu_2$) problem, under the restriction that all voter coordinates are in $\{-1,+1\}$ and $m=1$ (which is $\mathsf{NP}$-complete by Theorem \ref{thm:npcomplete} above). 

Given an instance $I_2 = (d,n,1,x,\mathbf{0},k)$ of OCP-D($\nu_2$), we construct the instance $I_p = (d,n,1,x',\mathbf{0},k)$ of OCP-D($\nu_p$), where for all $i \in [n]$, $x_i'$ is given by
\begin{equation}\label{eq:votercoords}
    (x_i')_j = \text{sign}((x_i)_j)\card{(x_i)_j}^{1/(p-1)}
\end{equation}
for all $j \in [d]$.

By Lemma \ref{lem:pballtangent}, the hyperplane $h_i$ at $\mathbf{0}$ tangent to the critical region of $x_i'$ (which is the $\ell^p$-norm ball centred at $x_i'$ with radius $\lVert x_i' \rVert_p$), is given by 
$h = \{y \in \mathbb{R}^d : \sum_{j \in [d]} \text{sign}((x_i')_j)\card{(x_i')_j}^{p-1} y_j = 0\} = \{y \in \mathbb{R}^d : \sum_{j \in [d]} p \cdot \text{sign}(x_i)_j\card{(x_i)_j} y_j = 0\} = \{y \in \mathbb{R}^d : \sum_{j \in [d]} (x_i)_jy_j = 0\}$.
Here, the second inequality follows from our definition of $x_i'$ and the last inequality points out that the coefficients of the hyperplanes at $\mathbf{0}$ tangent to the critical regions of the voters in our OCP-D($\nu_p$)-instance are indentical to the coefficients of the hyperplanes at $\mathbf{0}$ tangent to the critical regions in $I_2$.
From Lemma $\ref{L:COV-PLN}$, it then follows that for every set of voters $S \subseteq [n]$, there exists a point $t$ such that $V_2(t) = S$ under the instance $I_2$ if and only if there exists a point $t$ such that $V_p(t) = S$ under the instance $I_p$, which proves $\mathsf{NP}$-hardness.

Crucially, observe that if the voters coordinates in $I_2$ are unrestricted, then the reduction could involve potentially irrational-valued voter coordinates for $I_p$. However, note that we assumed that $(x_i)_j \in \{-1,+1\}$, so that (\ref{eq:votercoords}) actually boils down to $x_i'=x_i$, and hence $I_2$ is identical to $I_p$. 
\end{proof}

We remark that the containment in $\mathsf{NP}$ crucially relies on restricting to a variant of the problem where voter coordinates are in $\{-1,+1\}$. If they are not restricted as such, then the coordinates of the associated hyperplanes are potentially not rational-valued, and hence containment in $\mathsf{NP}$ is not clear. This gives rise to an open problem which we expect to be  challenging to settle: To illustrate this, consider the $\ell^p$ norm for $p = 3/2$: In order to determine whether a rational-valued location is in the halfspace corresponding to the critical region of a given voter, we would need to be able to solve the notorious \textsc{square root sum} problem \cite{orourke}, of which the complexity status is presently unknown.

The FLS$^{hb}$ problem is $\mathsf{NP}$-complete due to the number of dimensions $n'$ being unconstrained. When related to OCP via the reduction of Theorem \ref{thm:npcomplete}, $n'$ corresponds to $d$, the issue space dimension. In the next section, we will therefore focus on designing algorithms for OCP under the assumption that $d$ is a fixed constant. 

\subsection{Algorithms for solving OCP with bounded dimension}


The problems OCP-D($\nu$) and OCP-D($R$) are $\mathsf{NP}$ -complete when the number of issues $d$ is not bounded, but this is not the case when $d$ is fixed. Let $I$ be an election instance for which we want to solve the OCP problem. In this case, we can show for $m=1$ that the number of non-empty regions that we need to consider placing the new candidate in is bounded by a polynomial in $n$: Assume without loss of generality that the single candidate in our instance is located at the origin, i.e., $t_1 = \mathbf{0}$. Lemma \ref{L:COV-PLN} tells us that for any set of voters $S \subseteq [n]$, there exists a point intersecting with all balls $\{B_i\ :\  i \in S\}$ if and only if there exists a point in the intersection of $\{H_{i} : i \in S\}$. 


The number of regions in any hyperplane arrangement can be bounded as follows. 
\begin{lemma}\label{lem:numregions}
Let $\mathcal{H}$ be a central arrangement of $n\geq 2$ hyperplanes in $\mathbb{R}^d$. The number of regions of $\mathcal{H}$ (i.e., the number of connected components of $\mathbb{R}^d \setminus \mathcal{H}$) is at most
\begin{equation*}
    2 \sum_{j = 0}^{\min\{n-1,d-1\}} \binom{n-1}{j} .
\end{equation*}
This upper bound holds with equality when for every subset of $d$ hyperplanes of $\mathcal{H}$, its only intersection point is the origin.
\end{lemma}
We defer the proof to Appendix \ref{sec:numregions}. From the above, we can derive an algorithm that solves the OCP problem as follows. Given an election instance $(d,n,1,x,\mathbf{0})$, we compute each voter's corresponding hyperplane in the form of its vector of coefficients, which by Lemma \ref{lem:pballtangent} are given by $\text{sign}((x_i)_j)\card{(x_i)_j}^{p-1}$ for the $\ell^p$-norm (note that this is a rational number in case $p \in \mathbb{N}$). The algorithm then returns for each region $r$ in the central hyperplane arrangement the corresponding set of voters that would vote for the new candidate if they were placed in that region, close enough to the origin, as per Lemma \ref{L:COV-PLN}. 


Enumerating the entire set of regions can be done as follows. Given the $n$ hyperplanes $\mathcal{H} = \{h_1, \ldots, h_n\}$, we let a region in $\mathbb{R}^d \setminus \cup \mathcal{H}$ correspond to a $(+1,-1)$-vector $v_r \in \mathbb{R}^d$, where the $i$th coordinate indicates on which side of the $i$th hyperplane $r$ is located. 

We first find an initial non-empty region $r$ and its corresponding vector $v_r$. We then use the fact that a neighbouring region $r'$ of $r$'s has a vector $v_{r'}$ that can be obtained from $v_r$ by changing the sign of precisely one of its coordinates (this is due to the fact that neighbouring regions share an $(n-1)$-dimensional hyperplane as a facet in their boundaries). Thus, for each of the $n$ potential vectors $v'$ obtained by changing the sign of one of $v_r$'s coordinates, we may check which of those vectors correspond to a non-empty region of the hyperplane arrangement. Such a check can be carried out through employing any algorithm for deciding non-emptiness of a polytope described by a system of linear inequalities. Given that the dimension $d$ is fixed, we can employ the algorithm in \cite{megiddo}, which solves linear programs with $n$ inequalities and a fixed number of variables in time $O(n)$. 

Using these observations, we can iteratively explore the set of regions in the arrangement, keeping track of a list $D$ of $(+1,-1)$-vectors corresponding to regions that have been visited, along with a list of vectors $v_r$ corresponding to unexplored regions, i.e., which have been discovered as neighbours of visited regions. For each region $r$, the corresponding number of voters is then given by the number of $+1$-entries in $v_r$. There are $O(n^{d-1})$ regions to explore, by Lemma \ref{lem:numregions}, and each region requires solving $n$ linear programs (one non-emptiness check for each of its potential neighbouring regions), followed by checking for each of the discovered neighbouring regions whether it has already been discovered (which can be done in $O(n)$ time by representing $D$ as a prefix tree). Therefore, the time complexity of the resulting procedure is $O(n^{d+1})$. 

\begin{theorem}\label{thm:regionenumeration}
    There is a polynomial-time algorithm for OCP($\nu_p$) when $d$ is fixed, for $p \in \mathbb{N}, p > 1$, when the number of candidates $m$ is restricted to $1$. This algorithm runs in time $\Theta(n^{d + 1})$.
\end{theorem} 

Note that our algorithm described above, in its stated form, only finds a maximum-cardinality set $S$ of voters that can be acquired by a new candidate that positions themselves appropriately. The algorithm can be strengthened so that it returns in polynomial time an actual location in $\mathbb{R}^d$ that results in receiving the votes all voters in $S$: We first solve a linear program that finds a feasible point $y$ in $\cap_{i \in S} H_i$. We can then scale this point by a certain sufficiently small factor $\lambda$, which results in the desired point $y' = \lambda p$ that is sufficiently close to $\mathbf{0}$ so that we are guaranteed that $y' \in \cap_{i \in [n]} B_i$.

The following proposition gives us an appropriate scaling factor, and establishes that the scaling factor required has polynomial size in terms of the number of bits of the input instance, so that $y'$ in turn can be computed in polynomial time. Note that, given Lemma \ref{L:COV-PLN} and its proof, it might not be surprising to the reader that $y$ can be scaled down such that it results in a point in $\cap_{i \in [n]} B_i$. However, it is much less clear that the appropriate scaling factor requires only a polynomial number of bits and can thus be computed efficiently, for any $\ell^p$-norm with $p \in \mathbb{N}$, and this is what is shown in the following proposition. 

\begin{proposition}\label{prop:scaledown}
Let $(d,n,m,x,t)$ be an election instance and let $p \in \mathbb{N}$, $p \geq 2$, let $S \subseteq [n]$ and let $y$ be a point in $\cap\{H_i : i \in S\}$. Define $S^i_0 = \{j \in [d] : (x_i)_j = 0 \vee y_j = 0\}$ and $S^i_{-1} = \{j \in [d] : \text{sign}((x_i)_j = -\text{sign}(y_j) \not= 0\}$.
The point $y' = y \cdot \min\{1, \check{\lambda_i} : i \in S\}$
lies in $\cap\{B_i(\ell^p)\ :\ i \in S\}$, where $\check{\lambda_i} = \min\{1, \epsilon\card{(x_i)_j/y_j} : j \in [d]\setminus S^i_0\}$, $\epsilon = 1/2$ if $S^i_{-1} = \varnothing$ and $\epsilon = 2^{-\lceil \log_2(\sigma^i_d)\rceil - p - 1}$, and where  $\sigma^i_d$ is the denominator of the ratio 
\begin{equation*}
    \frac{\sum_{j \in [d]} y_j\card{(x_i)_j}^{p-1} \text{sign}((x_i)_j)}{\card{\sum_{j \in S^i_{-1}} y_j}(x_i)_j\card{^{p-1} \text{sign}((x_i)_j)}}
\end{equation*} evaluated as a rational number.
\end{proposition}
The proof can be found in Appendix \ref{apx:scaledown}. 
We reiterate here (from Section \ref{sec:relwork}) that various related algorithms have been proposed in previous works  \cite{chan_mechanism_2021,bremnertukey,liutukey,densesthemisphere,deliberative}, to which our problem setting indirectly translates. Those prior algorithms are all incompatible with non-Euclidean norms, and do not output an actual point for the candidate to position themselves on.

For the two-dimensional case ($d = 2$), it is possible to improve the running time of Theorem \ref{thm:regionenumeration} to $O(n\log n)$ through a simple radial sweep algorithm. We provide details in Appendix \ref{apx:2d_case}.  Furthermore, we remark that a very straightforward $2$-approximation algorithm is achievable for $m=1$ by choosing the better of an arbitrary point and its negation. We provide the details in Appendix \ref{apx:simpleapprox}.

\section{Competing against many candidates ($m > 1$)}\label{sec:manycandidates}
%
In this section, we briefly discuss the general case with more than one candidate. We have observed that when $m=1$, the best location to position oneself at is always a point close to the competitor. This does not hold anymore when $m > 1$: In this situation, the intersection of the maximum number of voters' critical regions may be located away from any of the competitors. 

\begin{example}
As an example of such an election instance, let $d=1$, $n=m=2$, two candidates located at points $0$ and $7$, and the voters are located at points $2$ and $5$. A new candidate should now position themselves in the open interval $(3,4)$, in which case they will win the votes of both voters.
\end{example}
One may think of this phenomenon as trying to attract a large set of voters who are not necessarily too satisfied with their current choice of candidate, as they could be rather distant from them in the issue space. If there are multiple such sets associated to different candidates, choosing a position somewhere in the centre of the respective candidates can result in a substantial number of voters changing their vote to the new candidate.

Solving this problem is equivalent to designing an algorithm that finds the highest-cardinality set of balls that has a non-empty intersection, when provided with an arrangement of balls in $\mathbb{R}^d$. 
Currently, algorithms are known for the case of the Euclidean ($\ell^2$) norm. For $\ell^2$, an algorithm was proposed that runs in time $O(n^{d+1})$, or $d^3$ when $n < d$ \cite{crama_complexity_1995}. 
However, this algorithm is suitable for the case of closed balls only, whereas in our problem the critical regions of the voters form open balls. Also, as pointed out in Section \ref{sec:relwork}, the time complexity of the algorithm of \cite{crama_complexity_1995} is established by a reference to an untraceable manuscript of C. Meyer, and we cannot verify its correctness at present. An alternative algorithm was given in \cite{cramaibaraki}, but this algorithm only computes the maximum cardinality index set of balls with non-empty intersection, and does not compute an actual location. 

Lushchakova \cite{lushchakova_geometric_2020} provides an algorithm for \emph{deciding} whether a given set of Euclidean closed balls in \( \mathbb{R}^d \) has a non-empty intersection. Her algorithm addresses only the \emph{decision} version of the problem, and works with \emph{closed} rather than \emph{open} balls. It also relies on a computational model in which basic arithmetic operations and square roots are assumed to take unit time, which leads to potentially irrational intermediate results and limits practical implementation.

In contrast, our approach tackles the \emph{optimisation} version of the problem for open balls, and is designed under the standard Turing model of computation. In Appendix~\ref{apx:balls3}, we present an extension of our algorithm that scales to multi-candidate configurations while preserving tractability in bit-complexity terms. The resulting method runs in time \( O(n^{2d-4}(nd^2 + d^3 + n^2 \log n) + nm) \), and can be interpreted as a constructive refinement of Lushchakova’s decision procedure.




With respect to approximation algorithms, we make the following observation for general $\ell^p$-norms, $1 < p < \infty$: For $j \in [m]$, let $S_j \subset [n]$ be a maximum set of voters in $\{i\ :\ i \in Cl(B_j)\}$ such that $\cap \{B_i\ :\ i \in S_j\} \not= \varnothing$, and let $j^* = \arg_j \max \{\card{S_j} : j \in [m]\}$. As the point $p$ that maximises $\nu_p(p)$ must satisfy $\nu_p(p) \leq \sum_{i \in [m]}\card{S_j}$,
selecting a point in $\cap\{B_i : i \in S_{j^*}\}$ yields an $m$-approximation to OCP($\nu_p$). Since by Theorem \ref{thm:regionenumeration} computing $\card{S_j}$ can be done in $O(n^{d+1})$ time, we can compute $p$ in $O(mn^{d+1})$ time.
\begin{proposition}
    \label{prop:mapprox}
    There exists a polynomial time $m$-approximation algorithm for OCP($\nu_p$) for $p \in \mathbb{N}, p > 1$. This algorithm runs in time $O(mn^{d+1})$.
\end{proposition}


Lastly, we note that above algorithms (ours and those in previous literature) all rely on the idea of enumerating exhaustively a set of potential solutions that grows exponentially in $d$.
This is necessary in the worst case, when voters can be configured arbitrarily over $\mathbb{R}^d$. However, empirical research has consistently demonstrated that real-world attitudes tend to exhibit strong correlations.

For example, \citet{Goren2025} and \citet{BaldassarriGelman2008} document substantial correlations between attitudes toward abortion and same-sex marriage in the US ($r \approx 0.66$), while \citet{CurticeEvans2019} report that attitudes towards Brexit and immigration in the UK are nearly interchangeable ($r \approx 0.72$). \citet{BaumannEtAl2021} show that even culturally unrelated topics (e.g., transgender bathrooms and border security) can align under ideological polarisation. Even when correlations are modest, as in \citet{Fesnic2021}, they still affect the geometry of the issue space.

Such findings suggest that real-world preference profiles may occupy a more structured and lower-complexity subset of $\mathbb{R}^d$ than the full combinatorial range, reducing the number of regions to consider from the theoretical maximum of $\binom{n}{d}$. Appendix~\ref{sec:ITR_ALL} presents an algorithm that exploits such regularities—via principal subspace projections, latent ideological dimensions, or clustering along known axes—to avoid full enumeration. Although auxiliary to our main theoretical results, this approach could improve scalability in practice.

A further variant is when the new candidate’s location is restricted to a feasible region. Our multi-candidate algorithm can handle this by adding a constraint ball~$B$ to the ball-intersection instance and requiring the maximum-cardinality intersecting set to include~$B$. Detailed study of such constrained-positioning models is left for future work.

\section{Discussion}

In this paper, our positive results relate to using the $\ell^p$-norm with $1 < p < \infty$ as the distance measure employed by the voters. Some questions about the choice of distance measure remain unexplored, and some gaps remain to be filled. In particular, is it possible to provide a similar set of results for the $\ell^1$ and $\ell^{\infty}$ norms? 
One may also consider the case of having voter-specific distance measures. For this direction, we can provide an initial negative result in Appendix \ref{sec:mixeddistance}.

For the many-candidate case of our problem, important open problems remain: In Section \ref{sec:manycandidates}, we have discussed three previously proposed algorithms that are closely related to our problem restricted to the Euclidean norm, but none of them solve our candidate positioning problem fully in a satisfactory way. 
We proposed a non-trivial adaptation of one of those algorithms, resulting in an algorithm that works correctly for the Euclidean norm and in polynomial time under the standard model of computation.
Our current algorithm runs in time roughly $O(n^{2d})$ and we wonder whether we can improve the factor $2$ in its exponent. We also leave open for future work the question whether a polynomial time algorithm exists for the many-candidate case for other distance measures, in particular general $\ell^p$-norm induced measures.


Many of our algorithms are $\mathsf{XP}$ in~$d$, and it remains open whether $\mathsf{FPT}$ algorithms exist or whether the base/exponent can be improved. It would also be interesting to establish conditional lower bounds under ETH, explore possible $\mathsf{W}[1]$-hardness, and derive inapproximability results for natural variants. Beyond complexity, promising directions include extending our framework to arbitrary positional rules (Appendix~\ref{apx:generalisations}), weighted voting where ballots have different weights, and range voting where scores depend on candidate distance. Game-theoretic variants, in which candidates have ``ideal'' positions and can move at a cost, could be analysed for iterative best-responses and equilibrium behaviour. Additional modelling dimensions include voters’ limited ability to distinguish very close candidates and hybrid settings balancing spatial proximity with independent score attributes (e.g., perceived competence). We also plan empirical experiments, to be reported in future work and in an extended journal version, to assess the practical performance of our algorithms.

\bibliography{references1.bib} 
\newpage
\clearpage
\appendix
    
\section{Proof of Lemma \ref{lem:numregions}.}
\label{sec:numregions}



\begin{proof}[Proof of Lemma \ref{lem:numregions}]
If $n < d$, the claimed bound simplifies to $2^n$, and is easily seen to hold true: Each region $R$ can be identified by a binary $n$-dimensional $(+1,-1)$-vector where the $i$th coordinate indicates on which side of the $i$th hyperplane $R$ lies. There are $2^n$ such vectors.

Assume therefore that $n \geq d$. Let $R$ be the set of regions induced by $\mathcal{H}$ Take any hyperplane $h \in \mathcal{H}$, which divides $\mathbb{R}^d$ in two regions $A$ and $B$. Observe that by symmetry, the number of regions of $R$ that lie in $A$ equals the number of regions of $R$ that lie in $B$. Thus, it suffices to bound the number of regions in $A$. To count those regions, let $h'$ be any $(d-1)$-dimensional hyperplane parallel to $h$ that lies entirely in $A$. The complement of the hyperplanes $\mathcal{H}\setminus \{h\}$ restricted to the $(d-1)$-dimensional plane $h'$ forms a set of $(d-1)$-dimensional regions $R'$, and each region of $R'$ in $A$ must lie in a separate region of $R$. Thus, $\card{R} \leq 2 \card{R'}$. As $h$ is $(d-1)$-dimensional, and $\mathcal{H}\setminus\{h\}$ induces a hyperplane arrangement on the $(d-1)$-dimensional space $h$, the number $\card{R'}$ is no larger than the maximum number of regions of an arrangement of $n-1$ hyperplanes in $\mathbb{R}^{d-1}$.

We know from Zaslavsky's Theorem \cite{zaslavsky1975} that the number of regions of a hyperplane arrangement $\mathcal{H'}$ is $\card{\chi_{\mathcal{H'}}(-1)}$: the absolute value of the characteristic function of $\mathcal{H'}$ evaluated at $-1$, given by
\begin{equation*}
    \chi_\mathcal{H'}(t) = \sum_{S \subseteq \mathcal{H'}\ :\ \cap S \not= \varnothing} (-1)^{\lvert S\rvert}t^{\text{dim}(\cap S)} .
\end{equation*}

In case for each subset of $d$ hyperplanes of $\mathcal{H}$, their only intersect point is at the origin, then in particular no $d$ hyperplanes of $\mathcal{H}\setminus \{h\}$ intersect at a point in $h$. In other words, the arrangement $\mathcal{H}\setminus\{h\}$ restricted to the subspace $h$ is in general position. For hyperplane arrangements of $n-1$ hyperplanes in general position in $\mathbb{R}^{d-1}$, it is straightforward to derive that the characteristic function (for $t = -1$) evaluates to
\begin{equation*}
    \card{\chi_{\mathcal{H}'}(-1)} = \sum_{i = 0}^{d-1}\binom{n-1}{i} ,
\end{equation*}
see, e.g., \cite{stanley2007hyperplane}. This establishes the second part of our claim.

Although it seems intuitively clear that the above should also be an \emph{upper bound} for the number of regions of a hyperplane arrangement in \emph{non-general position} (which is what remains to be shown), this is surprisingly not that straightforward to establish directly from the definition of the characteristic function itself. 

Instead, we reason as follows. Denote by $h_1, \ldots, h_{n-1}$ the hyperplanes in $\mathcal{H}\setminus\{h\}$. For $i \in [d-1]$ and $j \in [n-1]$ let $N(i,j)$ be the maximum number of regions of an arrangement of $j$ hyperplanes in $\mathbb{R}^{i}$. Consider the process of adding the hyperplanes to the arrangement one by one. When adding the last hyperplane $h_{n-1}$ to the arrangement of the first $n-2$ planes, there are a maximum of $N(d-1,n-2)$ regions before adding $h_{n-1}$ to the arrangement, and all regions through which $h$ passes are split into two regions. Since there is a correspondence between each $(d-2)$-dimensional region of $h_{n-1} \setminus \{h_1, \ldots, h_{n-2}\}$ and each $(d-1)$-dimensional region of the arrangement $\{h_1,\ldots, h_{n-2}\}$, the hyperplane $h_{n-1}$ passes through a maximum of $N(d-2,n-2)$ regions, so that adding $h_{n-2}$ to the arrangement results in $N(d-1,n-2) + N(d-2,n-2)$ regions. By induction, we may assume that $N(d-1,n-2) = \sum_{i=0}^{d-1}\binom{n-2}{i}$ and that $N(d-2,n-2) = \sum_{i=0}^{d-2} \binom{n-2}{i}$.

We then obtain the desired result:
\begin{align*}
    & \card{R} = 2\card{R'} \leq 2 ( N(d-1,n-2) + N(d-2,n-2) ) \\
    & \qquad = 2 \left( \binom{n-2}{0} + \sum_{i=1}^{d-1} \left(\binom{n-2}{i} + \binom{n-2}{i-1}\right)\right) \\
    & \qquad = 2\left(\binom{n-1}{0} + \sum_{i=1}^{d-1} \binom{n-1}{i}\right) .
\end{align*}
\end{proof}

\section{The hyperplane tangent to an $\ell^p$-Ball}\label{apx:pballtangent}
\begin{proof}[Proof of Lemma \ref{lem:pballtangent}]
Observe that the coefficients $\text{sign}(z_j)\card{z_j}^{p-1}, j\in [d]$ coincide with the gradient of the function $f : \mathbb{R}^{d} \rightarrow \mathbb{R}$ that expresses the $\ell^p$-norm distance to $x_i'$, at $\mathbf{0}$. The fact that this gradient forms the claimed coefficient vector can be derived through basic convex analysis tools (see \cite{rockafellarconvex}; we assume familiarity with some standard terminology in that field): We construct the supporting hyperplane $\bar{h}$ of the epigraph of the function $f$, at the point $(\mathbf{0}, f(\mathbf{0}))$, which is the hyperplane
\begin{equation*}
    \left\{y \in \mathbb{R}^{d+1}\ :\ \sum_{j \in [d]} \text{sign}(z_j)\card{z_j}^{p-1}y_j - y_{d+1} = -f(\mathbf{0})\right\} .
\end{equation*} 
Thus, for every point $y \in \mathbb{R}^{d+1}$ that is in the epigraph of $f$, we have $\sum_{j \in [d]} \text{sign}(z_j)\card{z_j}^{p-1}y_j - y_{d+1} \leq -f(\mathbf{0})$, and equality holds at the point $(\mathbf{0}, f(\mathbf{0}))$.
Then we compute intersection of the epigraph with the space $P = \{y \in \mathbb{R}^{d+1} : y_{d+1} = f(\mathbf{0})\}$ to obtain the convex set $S$, and observe that the supporting hyperplane $h$ restricted to the linear subspace $P$ is still a supporting hyperplane for $S$ relative to $P$, since both $S$ and $h$ contain $(\mathbf{0}, f(\mathbf{0}))$. Subhyperplane $h$ is formulated as follows. 
\begin{equation*}
    h = \left\{y \in \mathbb{R}^{d} \times \{f(\mathbf{0})\}\ :\ \sum_{j \in [d]} \text{sign}(z_j)\card{z_j}^{p-1}y_j = 0\right\} .
\end{equation*} 
Furthermore, $S$ can be formulated as follows.
\begin{equation*}
    S = \left\{y \in \mathbb{R}^{d} \times \{\mathbf{f(0)}\}\ :\ \sum_{j \in [d]} \card{y_j-z_j}^p \leq \sum_{j \in [d]} \card{z_j}^p\right\}.
\end{equation*}
Discarding the last coordinate altogether from the above two expressions, we obtain the claimed result: The convex set $S$ expresses the $\ell^p$-norm ball $B$, and $h$ expresses a supporting hyperplane tangent to it, with the claimed coefficients.
\end{proof}

\section{Bounding the Scaling Factor}\label{apx:scaledown}
\begin{proof}[Proof of Proposition \ref{prop:scaledown}]
We regard the norm $\ell^p$ as fixed, and write $B_i$ instead of $B_i(\ell^p)$ for notational convenience.

Consider an arbitrary $i \in S$. Observe that the half-line $\{\lambda y : \lambda \in [0,\infty)\}$ intersects the boundary of $B_i$, at two points, one of which is $\mathbf{0}$. All points on this half-line stricly in between those two points lie in $B_i$, and observe that by definition of $B_i$, such points $\lambda y$ satisfy that $\sum_{j \in [d]} \card{\lambda y_j-(x_i)_j}^p < \sum_{j \in [d]}\card{(x_i)_j}^p$.
Thus, the point $\lambda^* y$ that minimises $\sum_{j \in [d]} \card{\lambda^* y_j-(x_i)_j}^p$ must be in $B_i$, and all points $\lambda y$ such that $0 < \lambda < \lambda^* y$ are in $B_i$ as well.

We therefore consider the derivative $f(\lambda)$ of the above expression with respect to $\lambda$, given by $f(\lambda) = \sum_{j \in [d]} p y_j \card{\lambda y_j - (x_i)_j} \text{sign}(\lambda y_j - (x_i)_j)$. Note that $f$ is negative for $\lambda \in (0,\lambda^*]$, positive for $\lambda \in [\lambda^*,\infty)$ and $0$ at $\lambda^*$. We are interested in finding a choice of $\lambda$ such that $f(\lambda) < 0$, as then we are assured that $\lambda y \in B_i$.

We partition the index set $[d]$ into $S_{-1}$, $S_{+1}$ and $S_0$: The indices $j \in S_{-1}$ correspond to those where $\text{sign}((x_i)_j) = -\text{sign}(y_j) \not=0$, the indices $j \in S_{+1}$ satisfy that $\text{sign}((x_i)_j) = \text{sign}(y_j) \not=0$, and for the indices $j \in S_0 = [d]\setminus\{S_{-1},S_{+1}\}$ it holds that either $(x_i)_j = 0$ or $y_j = 0$. We can now write $f$ as
\begin{align*}
& f(\lambda) = p\left( \sum_{j \in S_{-1}} y_j\card{\lambda y_j - (x_i)_j}^{p-1}\text{sign}(\lambda y_j - (x_i)_j) \right. \\
& \qquad \left. + \sum_{j \in S_{+1}} y_j \card{\lambda y_j - (x_i)_j}^{p-1}\text{sign}(\lambda y_j - (x_i)_j) \right).
\end{align*}
We let $\check{\lambda}_i = \min\{\epsilon\card{(x_i)_j/y_j}\ :\ j \in S_{+1} \cup S_{-1}\}$, where $1 > \epsilon > 0$ is a small constant which we specify later, and we consider $f$ for $\lambda$ in the range $(0, \check{\lambda}_i]$. In this range observe that for $j \in S_{+1}$ we have $\text{sign}(\lambda y_j - (x_i)_j) = \text{sign}(\epsilon(x_i)_j - (x_i)_j) = \text{sign}(-\epsilon(x_i)_j) = -\text{sign}((x_i)_j)$. Furthermore, for $j \in S_{-1}$ we have $\text{sign}(\lambda y_j - (x_i)_j) = \text{sign}(\epsilon(-(x_i)_j/y_j) y_j - (x_i)_j) = -\text{sign}((x_i)_j)$.
Thus, for $\lambda \in (0,\check{\lambda}_i]$ we may write $f$ as follows.
\begin{align}
& f(\lambda) = -p \left( \sum_{j \in S_{-1}} y_j\card{\lambda y_j - (x_i)_j}^{p-1} \text{sign}((x_i)_j) \right. \notag \\
& \qquad \qquad + \left. \sum_{j \in S_{+1}} y_j \card{\lambda y_j - (x_i)_j}^{p-1} \text{sign}((x_i)_j)\right) \label{eq:partitioned}.
\end{align}
Observe that the terms in the first summation are negative, whereas the terms in the second summation are negative.
We subsequently derive (for $\lambda \in (0,\check{\lambda}_i]$) useful upper and lower bounds on the expression $\card{\lambda y_j - (x_i)_j}$ that occurs in all terms of the above summations, as follows: For $j \in S_{-1}$, 
observe that $\card{\lambda y_j - (x_i)_j} \leq \card{\check{\lambda}_i y_j - (x_i)_j} \leq \card{\epsilon(-(x_i)_j/y_j) y_j - (x_i)_j} \leq (1+\epsilon)\card{(x_i)_j}$.

For $j \in S_{+1}$, as $\lambda y_j$ is at most $\epsilon (x_i)_j/y_j \cdot y_j = \epsilon (x_i)_j \leq (x_i)_j$, we have $\card{\lambda y_j - (x_i)_j} \geq (1-\epsilon)\card{(x_i)_j}$.

These bounds imply that we can upper bound (\ref{eq:partitioned}) as follows, for $\lambda \in (0,\check{\lambda}_i]$: 
\begin{align}
& f(\lambda) \leq -p \left((1+\epsilon)^{p-1}\sum_{j \in S_{-1}} y_j\card{(x_i)_j}^{p-1} \text{sign}((x_i)_j) \right. \notag\\
& \qquad + \left. (1-\epsilon)^{p-1}\sum_{j \in S_{+1}} y_j\card{(x_i)_j}^{p-1} \text{sign}((x_i)_j) \right) \notag \\
& = -p \Bigg(((1+\epsilon)^{p-1} - (1-\epsilon)^{p+1}) \sum_{j \in S_{-1}} y_j\card{(x_i)_j}^{p-1} \text{sign}((x_i)_j)  \label{eq:negative} \\
& \qquad + \left. (1-\epsilon)^{p-1} \sum_{j \in [d]} y_j\card{(x_i)_j}^{p-1} \text{sign}((x_i)_j) \right). \label{eq:positive}
\end{align}
Let $\sigma^-$ denote the summation in line (\ref{eq:negative}) and let $\sigma^+$ denote the summation in line (\ref{eq:negative}).

Observe that $\sigma^+$ is a strictly positive quantity, as $y_i$ lies in $H_i$, which is described by the coefficients $\card{(x_i)_j}^{p-1} \text{sign}((x_i)_j)$ occurring in the summation. The quantity $\sigma^-$ on the contrary is non-positive by definition of $S_{-1}$, and is strictly negative is $S_{-1} \not=\varnothing$. To show that $f(\lambda)$ is negative for $\lambda \in (0,\check{\lambda}_i]$, we thus need to show that
\begin{equation*}
((1+\epsilon)^{p-1} - (1-\epsilon)^{p+1}) \card{\sigma^-} < (1-\epsilon)^{p-1}\sigma^+
\end{equation*}
which holds trivially in case $S_{-1} = \varnothing$, in which case $\epsilon$ can be chosen as an arbitrary value in $(0,1)$, say $1/2$. If $S_{-1} \not= \varnothing$ the above is equivalent to
\begin{equation*}
\frac{((1+\epsilon)^{p-1} - (1-\epsilon)^{p+1})}{(1-\epsilon)^{p-1}} < \frac{\sigma^+}{\card{\sigma^-}} .
\end{equation*}
Rewriting the left hand side as $((1+\epsilon)/(1 - \epsilon))^{p-1} - 1$ allows us to further rewrite the above into the equivalent expression
\begin{equation*}
\frac{(1+\epsilon)}{(1-\epsilon)} < \left(\frac{\sigma^i_n}{\card{\sigma^i_d}} + 1\right)^{1/(p-1)},
\end{equation*}
where $\sigma^i_n,\sigma^i_d$ are two integers such that $\sigma_n/\sigma_d$ is the result of evaluating the expression corresponding to $\frac{\sigma^+}{\card{\sigma^-}}$.
When setting $\epsilon = 2^{-\lceil \log_2(\sigma^i_d)\rceil - p - 1}$, the latter can be seen to hold, by the following derivation.
\begin{align*}
& \frac{(1+\epsilon)}{(1-\epsilon)} < (1+ \epsilon) = 1 + 2^{-\lceil \log_2(\sigma^i_d)\rceil - p-1} \\
& \qquad = ((1 + 2^{-\lceil \log_2(\sigma^i_d)\rceil-p-1})^{p-1})^{1/(p-1)} \\
& \qquad \leq (1 + 2^{p-1}2^{-\lceil \log_2(\sigma^i_d)\rceil-p-1})^{1/(p-1)} \\
& \qquad < \left(1 + \frac{1}{\sigma^i_d}\right)^{1/(p-1)} \leq \left(\frac{\sigma^i_n}{\sigma^i_d} + 1\right)^{1/(p-1)} ,
\end{align*}
where the second inequality holds due to the fact that $(1 + \epsilon)^{p-1} = \sum_{i = 0}^n\binom{k}{i} 1^{p-1-i}\epsilon^{i} = 1 + \sum_{i = 1}^n\binom{p-1}{i} 1^{p-1-i}\epsilon^{i} < 1 + 2^k\epsilon$.

This proves that for all $\lambda \in (0,\check{\lambda}_i]$ (where $\check{\lambda}_i = \min\{\epsilon\card{(x_i)_j/y_j} : j \in [d]\setminus S_0\}$, $\epsilon$ is set to $2^{-\lceil \log_2(\sigma^i_d)\rceil - p - 1}$ -- or $1/2$ in case $S_{-1} = \varnothing$, and $\sigma^i_d$ is the denominator of the expression $\sigma_+/\card{\sigma_-}$ evaluated to a rational number), the value $\lambda y$ lies in the ball $B_i$. Hence, the point $y \cdot \min \{\check{\lambda}_i\}$ lies in all balls $\{B_i : i \in S\}$, which proves our claim.
\end{proof}

\section{An efficient algorithm for one opponent and two dimensions}\label{apx:2d_case}
For the two-dimensional case ($d = 2$), it is possible to improve the running time of Theorem \ref{thm:regionenumeration} through a simple radial sweep algorithm, as there exists a much simpler approach to explore the regions of the hyperplane arrangement: Consider any two-dimensional election instance with one candidate, $(2,n,1,x,t)$, where $t = (0,0)$ and $x$ is a list of $n$ two-dimensional coordinates $x = (x_1, \ldots, x_n)$, where $x^i = ((x_i)_1, (x_i)_2)$ for $i \in [n]$. Assuming we use a $\ell^p$-norm, we compute each voter $i$'s hyperplane coefficients $(x_i')_1$ and $(x_i')^2$, given by $(x_i')_j = \text{sign}((x_i)_j)\card{(x_i)_j}^{p-1}$ (Lemma \ref{lem:pballtangent}).

The halfspace $H_i$ can be expressed as $\{(z_1,z_2)\ :\ (x_i')_1z_1 + (x_i')_2z_2 > 0\}$. 
Let $c_i = ((x_i')_1/-(x_i')_2)$, so that an equivalent way of expressing the latter halfspace is $H_i = \{(z_1,z_2)\ :\ c_iz_1 > z_2\}$.\footnote{For simplicity, here we assume that $(x_i')_2 \not= 0$ for all $i \in [n]$ to avoid division by $0$.} We now imagine the process of ``walking'' around a circle centred at the origin $(0,0)$ in a clockwise direction, starting at the point $(0,1)$. We may first compute in how many of the open halfspaces $(0,1)$ lies, and each time we cross a boundary of a halfspace, we add $+1$ or $-1$ to that number, depending on whether we enter or exit the halfspace. We keep record of the highest count that we encounter during our walk, along with the point at which that count was achieved.

During such a walk, each halfspace is encountered once during the first half of the trajectory (i.e., from $(0,1)$ to $(0,-1)$, at which point it is either entered or exited. During the second half of the trajectory (i.e., from $(0,-1)$ to $(0,1)$) the halfspaces are all encountered again, and in the same order. If a halfspace was entered during the first half of the walk, then it is exited during the second half of the walk, and vice versa. 
The order in which the corresponding lines are crossed during the first half (and thus also the second half) of the walk is straightforward: If $c_i > c_j$, then the boundary of $H_i$ is encountered before the boundary of $H_j$, and the sign of $(x_i')_2$ determines whether $H_i$ is entered or exited. If in the first half of the walk $(x_i')_2 > 0$, then the hyperplane $H_i$ is exited, and otherwise it is entered. We can translate the process we just described straightforwardly into a concrete algorithm, presented in Algorithm \ref{2d_case} (\textsc{OneOpponent\_2D}), where we also took into account the handling of voters with $(x_i')_2 = 0$.

The correctness and running time of Algorithm \ref{2d_case} follows from the formulation and discussion above. The sorting step at Line \ref{line:reindex} forms the computational bottleneck and runs in time $O(n\log n)$.

\begin{theorem}
    Algorithm 1 solves OCP($\nu_p$) in time $O(n \log n)$, when $m=1$, $d=2$, and $p \in \mathbb{N}, p > 1$.
\end{theorem}
Proposition \ref{prop:scaledown} can again be used, in conjunction with Algorithm \ref{2d_case}, to find an actual location corresponding to the output voter set.

\begin{algorithm}
  \caption{$d = 2, m = 1$ (i.e., 2 dimensions, one competing candidate).}\label{2d_case}
  \begin{algorithmic}[1]
    \Procedure{OneOpponent\_2D}{$n,x$}
    \State Let $(2,n,1,x,(0,0))$ be the input election instance. 
    Compute $x' = ( ((x_i')_1, (x_i')_2) : i \in [n] )$, where $(x_i')_j = \text{sign}((x_i)_j)\card{(x_i)_j}^{p-1}$ and let $X = \{i \in [n] : (x_i')_2 \not=0\}$. 
    \State Let $c_i = (x_i')_1/-(x_i')_2$ for all $i \in X$. Set $c_i = \infty$ for all $i \in [n]\setminus X$ for which $(x_i')_1 > 0$. Set $c_i = -\infty$ for all $i \in [n]\setminus X$ for which $(x_i')_1 < 0$.
    \State \label{line:reindex} Order the voters in descending order of $c_i$ (breaking ties arbitrarily). Denote by $\sigma(i)$ the $i$'th voter in this ordering.
    \State 
    Set $\textit{BestVoterSet}$ and $\textit{CurrentVoterSet}$ to the set of voters corresponding to the half-spaces that contain $(0,0)$. 
    \For{$i = 1,\ldots,n$} \Comment{{\tiny This loop refers to the first half of our walk (from $(\epsilon, 1)$ to $(0,-1)$). During our walk we encounter the voters/hyperplanes in ascending order of $c_i$.}} 
        \If{$(x_{\sigma(i)}')_2 < 0$}
        \Comment{{\tiny \emph{Entering} $\sigma(i)$'s halfspace.}}
            \State {\small $\textit{CurrentVoterSet} := \textit{CurrentVoterSet} \cup {\sigma^{-1}(i)}$.}
        \Else \Comment{{\tiny \emph{Exiting} $\sigma(i)$'s halfspace.}}
            \State {\small $\textit{CurrentVoterSet} := \textit{CurrentVoterSet} \setminus \sigma^{-1}(i)$.}
        \EndIf
        \If{$\card{\textit{CurrentVoterSet}} > \card{\textit{BestVoterSet}}$}
            \State {\small $\textit{BestVoterSet} := \textit{CurrentVoterSet}$}.
        \EndIf
    \EndFor
    \For{$i = n,\ldots,1$} \Comment{{\tiny This loop refers to the second half of our walk (from $(-\epsilon, -1)$ to $(0,1)$). During our walk we encounter the voters/hyperplanes in descending order of $c_i$.}} 
        \If{$(x_{\sigma(i)}')_2 < 0$}
        \Comment{{\tiny \emph{Exiting} $\sigma(i)$'s halfspace.}}
            \State {\small $\textit{CurrentVoterSet} := \textit{CurrentVoterSet} \setminus {\sigma^{-1}(i)}$.}
        \Else \Comment{{\tiny \emph{Entering} $\sigma(i)$'s halfspace.}}
            \State {\small $\textit{CurrentVoterSet} := \textit{CurrentVoterSet} \cup \sigma^{-1}(i)$.}
        \EndIf
        \If{$\card{\textit{CurrentVoterSet}} > \card{\textit{BestVoterSet}}$}
            \State {\small $\textit{BestVoterSet} := \textit{CurrentVoterSet}$.}
        \EndIf
    \EndFor
    \State \Return {\small $\textit{BestVoterSet}$, $y \in \cap\{H_i:i \in \textit{BestVoterSet}\}$.}
    \EndProcedure
  \end{algorithmic}
\end{algorithm}

\section{A Simple Approximation Algorithm}\label{apx:simpleapprox}
Here, we make very simple observation resulting in an approximation algorithm for $m=1$, if we suppose that the critical regions of the voters have a differentiable boundary (in the sense of satisfying the conditions of Lemma \ref{L:COV-PLN}). Let the candidate be located at $\mathbf{0}$ and consider the central hyperplane arrangement $\{h_1,\ldots, h_n\}$. Fix an arbitrary point $t \in \mathbb{R}^d$ that does not lie on any of $\{h_1, \ldots, h_n\}$. Let $S$ be the set of voters such that $t \in H_i$ for all $i \in S$. If $t$ is sufficiently close to $\mathbf{0}$ (which can be achieved by scaling down the initial choice of $t$ by an appropriate amount), by Lemma \ref{L:COV-PLN}, the point $t$ receives the votes by the voters in $S$. However, also observe that the point $-t$, i.e., on the opposite side of the candidate, receives all votes in $[n]\setminus S$. As a result, if there are less than $n/2$ voters who would vote for a candidate placed at $t$, there are more than $n/2$ voters who would vote for a candidate placed at $-t$. This gives us the following result. 
\begin{proposition}
For $m=1$, taking the better of the two locations $\{t,-t\}$,  where $t$ is an arbitrary point sufficiently close to $\mathbf{0}$ (and not on any of the hyperplanes $h_1,\ldots, h_n$) results in at least half the number of votes and thus yields a $2$-approximation to  OCP($\nu_p$). 
\end{proposition}
Observe that from the above simple insights it also follows that, in light of the OCP($R_p$) variant of the problem, it is usually possible for the new candidate to obtain the majority of the votes: If there is a location $t$ that receives less than half of the votes, then the location $-t$ must receive more than half of the votes. Thus, the only election instances where a new candidate cannot obtain a majority are those where every point $t$ is in exactly half of the halfspaces, and that only happens if for every voter $i$ (located at $x_i$), there is another voter at location $-x_i$, which can be checked in linear time. Note that, nonetheless, finding an actual location that results in receiving a majority vote requires in general more effort than simply picking an arbitrary location $t$ and choosing between $t$ and $-t$ (but a radial sweep as done in Algorithm \ref{2d_case} suffices for this). 

\section{An algorithm for competing against many candidates in fixed dimension}\label{apx:balls3}

Crama et al. \cite{crama_complexity_1995} and \cite{cramaibaraki} present algorithms that find an maximum cardinality set of non-intersecting balls in an arrangement of Euclidean balls in $\mathbb{R}^d$, which is what we need for the general case of our problem where there is an arbitrary number of candidates $m$ present in the election instance. The claimed time complexity for the algorithms are $O(n^{d+1})$ and $O(n^{d+2})$, under the assumptions that the balls are represented by their centres and radii. A crucial step in the algorithm of \cite{crama_complexity_1995} requires the computation of the solution to a set of equations, one of which is quadratic and the remainder of which are linear. The authors refer to an untraceable dissertation by C. Meyer for an algorithm and time complexity analysis of solving such a system. It is thus at present unclear to us whether the authors' time complexity claim is valid and whether the authors' claim on it holds under a classical model of computation. The algorithm of \cite{cramaibaraki} relies on different and interesting techniques, and its analysis and specification is complete. However, it has the drawback that it does not find an actual point in $\mathbb{R}^d$, but merely outputs the maximum index set of balls with a non-empty intersection. Another aspect of the algorithms in \cite{crama_complexity_1995} that is not compatible with our setting is that in their algorithm the balls are assumed to be closed, whereas our election setting translates to arrangements of open balls instead.

To partially address this issue, in this appendix, we provide an alternative solution, based on an algorithm by Lushchakova \cite{lushchakova_geometric_2020}. We prove that this adaptation of Lushchakova's work has a polynomial running time when $d$ is fixed, where it's time complexity is $O(n^{2d-4}(nd^2 + d^3 + n^2\log n) + nm)$. 

Concretely, \cite{lushchakova_geometric_2020} provides an algorithm $\textit{BALLS3}$, that decides whether a set of $n$ Euclidean balls in $\mathbb{R}^d$ have a common intersection point. Note that $\textit{BALLS3}$ solves a \emph{decision} problem closely related to our optimisation problem OCP($\nu_2$), but some additional work is needed to make the algorithm solve the corresponding maximisation problem, which is what we need. Also, like \cite{crama_complexity_1995}, Lushchakova's algorithm assumes an arrangement of closed balls as input, so we need to make the appropriate modifications to make the algorithm compatible with open balls instead.

Another important problem that we resolve about Lushchakova's method is that her original algorithm runs on a rather powerful version of a real RAM model, in which both the elementary arithmetic operations and the square root operation can be carried out in unit time. This issue is unaddressed in Lushchakova's paper \cite{lushchakova_geometric_2020}. At the end of this section, we show how to turn the algorithm into one that runs in strictly polynomial time, in the bit complexity sense, by making adaptation so that in particular the need for taking square roots is eliminated.

\subsection{A modified \textit{BALLS3} algorithm}
\label{sec:modball3}

The original $\textit{BALLS3}$ algorithm \cite{lushchakova_geometric_2020} for $d > 2$ takes a recursive approach to solve the problem of deciding whether a given set of balls (each described by a radius and a centre point), have a common intersection point, where an auxiliary algorithm $\textit{BALLS1}$ (also derived independently in \cite{drezner}) is provided to handle a ``base case'', where $d=2$. We adapt these algorithms to solving the OCP($\nu_2$) problem instead. Note that besides the fact that our problem is a maximisation problem rather than a decision problem, the input is also of a different form: The radii of the critical regions (i.e., balls) of the voters are not given as part of the input, so we will pre-compute these as a first step. A further difference is that the original $\textit{BALLS1}$ and $\textit{BALLS3}$ algorithms assume closed balls, and may output points that may lie at the boundaries of some of the balls. In our setting, we must output points that lie in the interior instead. This translates into various subtle but important modifications to $\textit{BALLS1}$, which carry over to $\text{BALLS3}$, as it uses $\textit{BALLS1}$ as a subroutine. Our modification of $\textit{BALLS1}$ is as follows.

\paragraph*{$\textit{ModifiedBALLS1}$:}
\begin{enumerate}
    \item Let $(2,m,n,x,t)$ be the input election instance (i.e., where $d = 2$). Compute the radius $r_i$ of $B_i$ for all $i\in[n]$.
    \item For all pairs $i,i'\in[n]$ of voters, determine whether $B_i$ and $B_{i'}$ intersect. If they do, record the two points where the boundaries intersect by $p_{i,i'}^1$ and $p_{i,i'}^2$.
    \item For each $i \in [n]$, sort the set of all points $P_i = \{p_{i,i'}^1, p_{i,i'}^2\ :\ j \in [n]\setminus\{i\}, B_i \cap B_{i'} \not= \varnothing \}$ on the boundary of $B_i$ in clockwise order (starting from an arbitrary point $p^i_0 \in P_i$). Refer to the resulting ordered points as $(p^i_0, p^i_1, p^i_2, \ldots, p^i_{\card{P_i}})$ and (for notational convenience) define $p^i_{-1}$ and $p^i_{\card{P_i} + 1}$ as $p^i_{0}$.
    \item For each $i \in [n]$, for each $k = 0, \ldots, \card{P_i}$, compute the set $S(i,k) = \{i' \in [n] : \text{mid}(p_k^i,p_{k+1}^i) \in B_{i'}\}$ where $\text{mid}(p_k^i,p_{k+1}^i)$ is an arbitrary point on the open segment $\text{Bnd}(i,k)$ of the boundary that runs clockwise from $p_k^i$ to $p_{k+1}^i$ (observe that, if $\text{Bnd}(i,k)$ is non-empty, this is well-defined as $S(i,k)$ is independent of the choice of point on the boundary segment), and $\text{mid}(p_k^i,p_{k+1}^i) = p_k^i$ if $\text{Bnd}(i,k)$ is empty. Note that after computing $S(i,0)$, the sets $S(i,k)$ can be efficiently computed by letting $i'$ be the voter such that $p_k^{i} \in \{p_{i,i'}^1,p_{i,i'}^2\}$, setting $S(i,k) = S(i,k-1) \cup \{i'\}$ if $i \in S(i,k-1)$ and setting $S(i,k) = S(i,k-1) \setminus \{i'\}$ otherwise.
    \item Let $k,i$ be a pair maximising $\card{S(i,k)}$ among all $\{S(i,k)\ :i \in [n], 0 \leq k \leq \card{P_i}\}$. Output the set $S(i,k)$ along with any point in the open line segment $(p_k^i,p_{k+1}^i)$.
\end{enumerate}

\begin{proposition}
$\textit{ModifiedBALLS1}$ solves OCP($\nu_2$) and runs in time $O(n^2\log n + nm)$.
\end{proposition}
\begin{proof}[Proof]
The correctness follows from the analysis in \cite{lushchakova_geometric_2020}, with the following additional remarks taken into account (which account for the modifications made to the original algorithm). The sets $S(i,k)$ represent the voters who would vote for a point on the segment $\text{Bnd}(i.k)$ of the boundary of $B_i$, and this set changes by only a single voter for each consecutive segment $\text{Bnd}(i,k)$ of the boundary of $B_i$. It is clear that the set $S^*(i,k)$ maximising $\card{S_{i,k}}$ over $0 \leq k \leq \card{P_i}$ is the maximum set of voters excluding $i$, such that their critical regions all intersect and furhermore intersect $B_i$. The algorithm outputs the optimal solution because it outputs the set $S^*(i,k) \cup \{i\}$ of maximum cardinality among $\{S^*(i,k) : i \in [n]\}$ along with a point in the intersection of all critical regions of the output set of voters. To see that the point output is indeed in all of those critical regions, we note the following: Since the intersection points $p_0^i,\ldots,p_{\card{P_i}}^i$ are points on the boundary of $B_i$, and our definition of the critical region of a voter does not include the boundary, no point of $\text{Bnd}(i.k)$ is in $B_i$ itself. However, clearly, the open line segment $(p_k^i,p_k^{i+1})$ does include $B_i$, and must lie in $\cap \{B_i\ : S(i,k)\}$ due to the fact that $p_k^i$ and $p_k^{i+1}$ are adjacent intersection points on the boundary of $B_i$, and because $\cap \{B_i\ : S(i,k)\}$ is convex. Furthermore, note that some of the steps in the above description implicitly make use of the radii computed in the first step.

The runtime of the algorithm is $O(n^2 \log n + mn)$, where the $O(nm)$ stems from the computation of the radii at the start. And the $O(n^2 \log n)$ stems from the $n$ sorting operations in the third step of the algorithm. Note that Step 4 can be implemented to run in time $O(n^2)$ because the algorithm does not need to record all of the values of $S(i,k)$: rather, it suffices to record only the $S(i,k)$ with the highest cardinality.
\end{proof}

Next, we state our modified version of $\textit{BALLS3}$. Given an election instance $I = (d,n,m,x,t)$, the algorithm and its correctness is based on the following observation, where we use the notation $\mathcal A^K(S)$ (see Definition \ref{def:keynotions}). Furthermore, we say that a voter in a set $S$ is \emph{boundary-defining for $S$} iff $\mathcal A^K(S \setminus \{i\}) \not= \mathcal A^K(S)$. The algorithm uses as a key idea the \emph{hyperplane that contains the intersection of the boundaries of $B_i$ and $B_{i'}$}, where $i$ and $i'$ are two voters whose critical regions intersect and are not contained in each other. We call this hyperplane $h_{i,i'}$. The intersection of the boundaries of $B_i$ and $B_{i'}$ is in itself a sphere, and thus forms the boundary of a ball of one dimension less, that lies on the hyperplane $h_{i,i'}$.

\begin{lemma}\label{lem:spherehyperplane}
If the maximum-cardinality set of voters $S$ for which $\mathcal A^K(S)$ is non-empty has at least two boundary-defining voters $i$ and $i'$, then $h_{i,i'} \cap \mathcal A^K(S) \not=\varnothing$.
\end{lemma}
\begin{proof}
    Let $i,i'\in S$ be two voters that are boundary-defining for $S$. 
    $B_{i} \cap B_{i'}$ has a non-empty intersection with the halfspace on either side of $h_{i,i'}$. Moreover, the part $P_{i}$ of the boundary of $B_{i} \cap B_{i'}$ that is part of $B_i$'s boundary lies on a different side of $h_{i,i'}$ than the part $P_i'$ of the boundary of $B_{i} \cap B_{i'}$ that $B_{i'}$ contributes to. Because $i$ and $i'$ are boundary-defining for $\mathcal A^K(S)$, the boundary of $\mathcal A^K(S)$ has a non-empty intersection with $P_{i}$ and $P_{i'}$, and hence has a non-empty intersection with the halfspaces on either side of $h_{i,i'}$. Because $\mathcal A^K(S)$ is convex, it must thus also intersect with $h_{i,i'}$. 
\end{proof}

\begin{corollary}
\label{cor:hyperplanes}
     If the maximum-cardinality set of voters $S$ for which $\mathcal A^K(S) \not=\varnothing$ has at least two boundary-defining voters, then there is a point $t'$ maximising $\card{\{i \in [n] : t \in B_i\}}$ that lies in $\cup \{h_{i,i'} : i,i' \in [n], B_i \cap B_{i'} \not= \varnothing\}$.
\end{corollary}
The below algorithm $\textit{ModifiedBALLS3}$ takes as input either an election instance, or a representation of such an instance represented as a collection of $d$-dimensional balls (given by their centre points and radii). The algorithm uses the above corollary to find the best points on the hyperplanes $h_{i,i'}$ and considers these hyperplanes one by one. When restricted to such a hyperplane, the instance can be reduced to an instance of one dimension less, which leads to a recursive approach.


\paragraph*{$\textit{ModifiedBALLS3}$}
\begin{enumerate}
    \item \label{step:1} Let $I$ be the input instance, either given as an election instance or as a tuple $(d \in \mathbb{N}, n \in \mathbb{N}, (x_i,r_i)_{i \in [n]})$ where $(x_i,r_i)$ describes a ball $B_i$, given by its centre point $x_i \in \mathbb{R}^d$ and radius $r_i \in \mathbb{R}$. If $I$ is given in the form of an election instance, we first pre-compute the radius $r_i$ of $B_i$ for all $i\in[n]$.
    \item \label{step:2} Compute the set $P = \{(i,i')\ :\ i \not= i, B_i \cap B_{i'} \not= \varnothing\}$ of voter-pairs whose critical regions intersect. 
    \item \label{step:3} For every $(i,i') \in P$:
    \begin{enumerate} 
        \item \label{step:3a} compute the hyperplane $h_{i,i'}$, which contains the intersection of $B_i$'s and $B_{i'}$'s spheres.
        \item \label{step:3b} For the simplicity of the exposition, rotate the instance through a linear transformation $T$, such that the normal vector of $h_{i,i'}$ becomes $(0,0,...,0,1)$. Denote the resulting instance by $I'$ and the rotated hyperplane by $h_{i,i'}'$. 
        \item \label{step:3c} Project $I$ on $h_{i,i'}'$ as follows, resulting in a new instance $I'$ of dimension $d-1$, with balls $\{(x_i',r_i')\ :\ B_i \cap h_{i,i'}' \not= \varnothing\}$, given by their centre points and radii. For all $i'' \in [n] \setminus\{i'\}$ for which $B_i \cap h_{i,i'}' \not= \varnothing$, we set $x_{i''}' = ((x_{i''})_1, \ldots, (x_{i''})_{d-1})$ (i.e., we drop the last coordinate of $x_{i''}$). We set the radius $r_{i''}$ to be the radius of the $(d-1)$-dimensional ball $B_{i''} \cap h_{i,i'}$, which is given by $r_{i''} = \sqrt{r_{i''}^2 - \text{dist}(x_{i''},h_{i,i'})^2}$, where $\text{dist}(x_{i''},h_{i,i'})$ is the distance from $x_{i''}$ to $h_{i,i'}$.
        \item \label{step:3d} If $d-1 = 2$, call $\textit{ModifiedBALLS1}(I')$. Otherwise, call $\textit{ModifiedBALLS3}(I')$. Let the result of the call (rotated back to its original coordinates via $T^{-1}$) be the point $t_{i,i'}$ and voter set $S_{i,i'}$.
    \end{enumerate}    
    \item \label{step:4} Let $(i_1,i_2) \in \arg_{(i,i')}\max \{\card{S_{i,i'}}\ :\ (i,i') \in P\}$. 
    \item \label{step:5} Furthermore, for each $i \in [n]$, let $t_i$ be an arbitrary point in $B_i$ and let $S_i = \{i' \in [n]\ :\ B_{i'} \supseteq B_i\}$. Let $i_{\max} = \arg_i \max \{\card{S_i}\ :\ i \in [n]\}$
    \item \label{step:6} If $\card{S_{i_{\max}}} > \card{S_{i_1,i_2}}$, return $S_{i_{\max}}$ together with the point $t_{i_{\max}}$. Otherwise return $S_{i_1,i_2}$ together with the point $t_{i_1,i_2}$. 
\end{enumerate}

\begin{proposition}
$\textit{ModifiedBALLS3}$ solves OCP($\nu_2$) and runs in time $O(n^{2d-4}(nd^2 + d^3 + n^2\log n) + nm)$.
\end{proposition}
\begin{proof}[Proof]
Correctness follows from the discussion above and from Corollary \ref{cor:hyperplanes}: $\textit{ModifiedBALLS3}$, for every $(i,i') \in P$, recursively calls itself on an instance of a lower dimension that results in an optimal location when the choice of location is restricted to the hyperplane $h_{i,i'}$. Among all the locations returned (one per $(i,i') \in P$), the algorithm selects the best one in Step \ref{step:4}. This will be the point that is output in Step \ref{step:6}, except when the premise of Corollary \ref{cor:hyperplanes} is not satisfied, i.e., in the case that for the maximum-cardinality set $S \subseteq [n]$ there are less than two boundary defining voters. The latter can only be the case when the boundary of $\mathcal A^K(S)$ coincides with the boundary of one of the balls in $S$, and in that case Step \ref{step:5} of the algorithm makes sure that an appropriate point is output.

The details on how to implement the computation of the hyperplanes $h_{i,i'}$ and transformation $T$ and the runtime analysis are identical to those of the original $\textit{BALLS3}$ algorithm, and we refer to  \cite{lushchakova_geometric_2020} for the details. In short, step \ref{step:1} is responsible for the $O(nm)$ term in the runtime bound. Step \ref{step:2} takes $O(n^2d)$ time. The loop of Step \ref{step:3} runs at most $O(n^2)$ times. Within one iteration of the loop, computation of $h_{i,i'}$ is done in $O(d)$ time. Construction of $T$ and $T^{-1}$ is done in $O(d^3)$ time, and applying $T$ to all balls in Step \ref{step:3b} is done in $O(nd^2)$ time. Computing the projected centre points and radii in Step \ref{step:3c} takes $O(n)$ time. Steps \ref{step:4} and \ref{step:5} take $O(n^2)$ and $O(n)$ time respectively. Altogether, we arrive at a computation time for the entire execution of the algorithm of $T(n,d) = O(nm + n^3d^2 + n^2d^3) + n^2T(n-1,d-1)$, matching the runtime of the original \textit{BALLS3} algorithm asymptotically. This expression can be bounded by $O(n^{2d-4}(nd^2 + d^3 + n^2\log n)$.
\end{proof}

The matrix $T$ in $\textit{BALLS3}$ represents an orthoghonal transformation that can be formed by taking $d$ linearly independent points on the hyperplane $h_{i,i'}$ and subsequently applying Gram-Schmidt orthonormalisation. This latter procedure will yield irrational entries (square roots) in $T$. In order to ensure that all computations can be carried out on a Turing Machine or RAM, we note that the rotation operation can be entirely omitted from the algorithm: We may let all points reside in $\mathbb{R}^d$ throughout the entire computation, and just keep track of the affine subspace of $\mathbb{R}^d$ in which all the balls are located throughout each level of the recursion. At recursion depth $k \in \{0,\ldots,d-2\}$, this affine subspace is $(n-k)$-dimensional, and given by the intersection of $k$ hyperplanes $h_{i_1,i'_1},\ldots, h_{i_k,i_k'}$ for a collection $C$ of pairs $(i_1,i_1'),\ldots,(i_k,i_k') \in P$ (which each correspond to the pair $(i,i') \in P$ over which the iteration of Step 3 runs. Thus, the affine subspace of $\mathbb{R}^d$ in which the balls of the instance are located can be represented by the hyperplanes $h_{i_1,i'_1},\ldots, h_{i_k,i_k'}$ (given in terms of their normal coefficients). Step 3b of $\textit{BALLS3}$ can then be removed, and in Step 3d the multiplication with $T^{-1}$ can be omitted. Instead, in Step 3c, $x_i$ now needs to be projected on the (non-transformed) hyperplane $h_{i,i'}$. Denoting the coefficients of $h_{i,i'}$ by $(a_1,\ldots, a_d,b)$, this is done by setting $x_i' = x + \lambda a$, where $\lambda = (b-\sum_{j \in [d]} a_j(x_i)_j)/\lVert a \rVert_2^2$.
One further detail in the above formulation of $\textit{BALLS3}$ that we need to modify is the way the radii of the balls are represented: By keeping track of the \emph{square} of the radius of each ball, the square root computation in Step 3c can then be avoided, so that all values we work with remain rational numbers.

Altogether, this removes the need to compute any square roots and keeps the arithmetic confined to the elementary operations of addition, subtraction, division, and multiplication. The bit complexity of the algorithm remains limited due to the fact that the recursion depth is bounded by $d$: Using the fact that each elementary arithmetic operation on two rational numbers results in a rational where of a constant factor times (roughly double) the bit length of the largest of its two operands, the coefficients used to represent the centres and the radii at depth $d$ require a number of bits of at most $(Kd)^d$ times the largest coefficient in the input coordinates in both their numerator and denominator, where $K$ is a constant. This is due to the $d$-fold successive projection of the centres of the balls on lower-dimensional hyperplanes, and the successive squaring of the squares in the expressions of the radii. These modifications result in a strictly polynomial time algorithm on the standard models of computation.

\section{An observation on using mixed distance measures}\label{sec:mixeddistance}
We've shown that when $m=1$, and all voters use a convex and differentiable norm as their distance measure (in particular the $\ell^p$ norm for $1 < p < \infty$, the number of intersection regions $\card{\{S \subseteq [n]\ :\ \mathcal A^K(S)\not=\varnothing\}}$ to consider is bounded by $O(n^d)$ (Lemma \ref{lem:numregions}). This shows that the number of regions does not grow exponentially in $n$ for certain natural restrictions of the problem. A natural question to consider is whether this still holds when the distance measure varies per voter. The answer to this is generally negative: Figure \ref{fig:convex-expo}, shows an example of $n$ convex sets, which we can imagine to be the critical regions $B_1, \ldots, B_n$ of a set of $n$ voters in an election instance.\footnote{This example is actually a direct solution to Exercise 1.22 of the textbook \cite{graham_concrete_1994}.}
For this instance, the set $Y$ contains every subset of $[n]$.
The construction is as follows. We start with $2^n$ equidistant points on a circle $C$. Number the points from $1$ to $2^n$ and consider the binary representations of these numbers. We define the region $B_i, i \in [n]$ by defining its boundary, which is given by the union of
\begin{itemize}
\item the line segments from point $j$ to $j+1~~(\text{mod } 2^n)$ for all $j$ where the $i$th bit is $1$. 
\item the (clockwise) segment of $C$ from point $j$ to $j+1 ~~(\text{mod } 2^n)$ for all $j$ where the $i$th bit of $j$ is $0$.
\end{itemize}
Note that this construction yields a set of regions with a non-differentiable boundary. However, it is straightforward to modify the construction by smoothening out the corner points, so that the regions all have a differentiable boundary.

\begin{figure}[H]
    \centering
    \includegraphics[scale = 0.25]{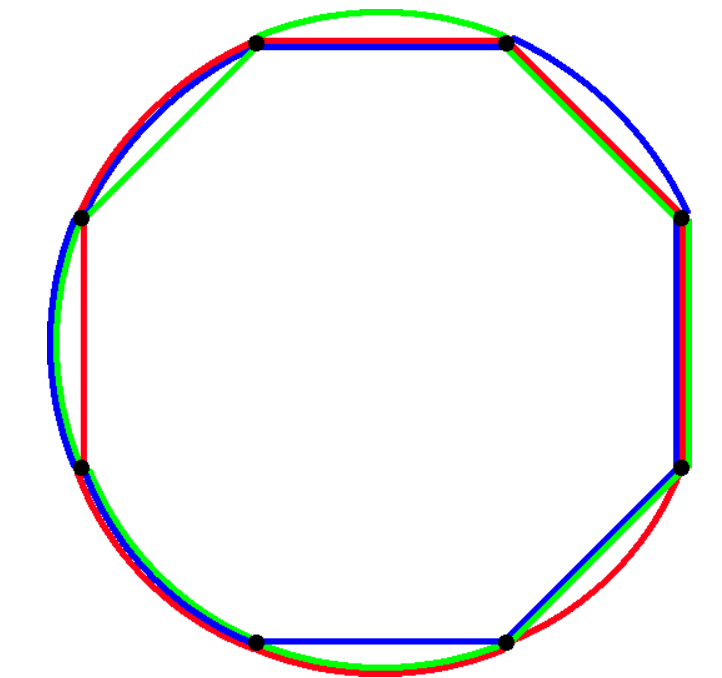}
    \caption{Construction a set of $n$ convex regions $\mathcal{B} = B_1,\ldots,B_n$ such that $\card{\{\cap S : S \subset \mathcal{B}\}} = 2^n$.}
    \label{fig:convex-expo}
\end{figure}



\section{Other Election Models}\label{apx:generalisations}
An underlying assumption in our spatial voting model has been that a plurality voting system is used, i.e., where each voter only votes for its most preferred (i.e., closest, according to an $\ell^p$ norm distance measure) candidate. We can straightforwardly  generalise our model to be compatible with arbitrary positional voting systems as well. In a positional voting system, each voter $i \in [n]$ assigns a certain number of $q_{ij}$ points to their $j$'th most preferred candidate, such that $q_{ij} \geq q_{i(j+1)}$. In our spatial model, a voter's $j$'th most preferred candidate corresponds to their $j$'th closest candidate with respect to the given distance measure. 

Important special cases of positional voting systems are
\begin{itemize}
    \item \emph{$k$-approval}, where $q_{ij} = 1$ for $j \in [k]$ and $q_{ij} = 0$ for $j \in \{k+1,\ldots,m\}$. Observe that $k=1$ corresponds to plurality voting and $k= m-1$ corresponds to a veto voting system.
    \item \emph{Borda}, where $q_{ij} = m-j$.
\end{itemize} 

We can now define a \emph{generalised election instance} as a tuple $(I,q)$, where $I$ is an election instance and $q \in \mathbb{N}^{n \times m}$ is a given matrix defining the scoring system. 

The problem of placing a new candidate in such an election on a position that receives the highest total score from all voters combined is still inherently of the same geometric nature as the one studied in the previous sections: We consider an arrangement of $\ell^p$-norm balls in $\mathbb{R}^d$ positioned as follows. We let $t_i^{(j)}$ denote the location of the $j$'th-closest candidate to $i$. For each $i \in [n]$ and $j \in [m]$, there is a number of $\hat{q}_{ij} = q_{ij} - q_{i(j+1)}$ balls $B_{i,j,1}(p),\ldots B_{i,j,\hat{q}_{ij}}(p)$ with radius $\lVert t_i^{(j)} - x_i \rVert_p$ at position $x_i$ (where we define $q_{i(m+1)}=0$). Under this definition of the ball arrangement, it can be readily verified that when the candidate chooses position $p \in \mathbb{R}^d$, the total score received is equal to the number of balls in which $p$ is contained.

Thus, the natural generalisation of the OCP($U$) problem where for $U$ we consider the function $SC_p(t) = \card{\{(i,j,k) \in [n] \times [m] \times [\hat{q}_{ij}]\ :\ t \in B_{i,j,k}(p)\}}$ is very similar to the original problem we studied above, apart from the fact that the generalised family of instances correspond to arrangements of balls that are potentially nested, whereas previously such nesting was not possible. The generalisation of $R_p$ can similarly be considered in this setting, although in this general setting, $R_p$ is not as straightforward to formulate as $SC_p$.

Altogether, this means that our results from Section \ref{sec:manycandidates} carry over, since the algorithms discussed there do not rely on the fact that the induced arrangement of balls are not contained in each other. Furthermore, of course, our results for the single-candidate case carry over trivially to this positional voting setting.

\newcommand{\RR}{\mathbb{R}}

\section{Iterate through All Areas}\label{sec:ITR_ALL}

Another approach to solving the OCP problem is to iterate through the set of non-empty regions in \(\mathbb{R}^d\) corresponding to a specific set of voters. This leads to the notion of a \emph{meaningful area}.

As every point \(p\) in the \(d\)-dimensional space either belongs or does not belong to each ball \(B_i\), we can collect the indices of the balls containing \(p\) into a set \(S \subseteq [n]\). That is, \(p \in B_i\) for all \(i \in S\), and \(p \notin B_j\) for all \(j \in [n]\setminus S\). We call the region satisfying these membership requirements the \emph{meaningful area} for \(S\):

\begin{definition}[Meaningful Area]
  Let \(S \subseteq [n]\). Then
  \[
     \mathcal A^K(S) \;=\; \Bigl\{p \in \mathbb{R}^d :\,
         p \in \bigcap_{i \in S} B_i \quad\text{and}\quad
         p \notin \bigcup_{j \in [n]\setminus S} B_j
     \Bigr\}.
  \]
  In words, \(\mathcal A^K(S)\) is the set of points that are contained in exactly those balls indexed by \(S\).
\end{definition}

In an election setting, each \(B_i\) represents a candidate's preference set. If a candidate places their policy \(t\) in \(\mathcal A^K(S)\), they receive \(\absl{S}\) votes (i.e.\ the number of voters corresponding to \(S\)). Thus, one might consider iterating through all non-empty \(\mathcal A^K(S)\) to find the best policy location. Suppose the space is divided into regions \(\mathcal A^K(_1), \mathcal A^K(_2), \dots\) by balls \(B_1\) to \(B_k\). If we add a new ball \(B_{k+1}\), we check the intersections of each existing \(\mathcal A^K(_i)\) with \(B_{k+1}\).

\subsection{Intermediate Areas and the Area-Updating Method}

To describe these regions during the iteration process, define the \emph{intermediate area} \(\mathcal A^K(^k)\) as
\[
   \mathcal A^K(S^k) \;=\; \Bigl\{p \in \mathbb{R}^d :\, p \in \bigcap_{i \in S} B_i \;\;\text{and}\;\; p \notin \bigcup_{j \in [k]\setminus S} B_j \Bigr\}.
\]
For simplicity, let \(\mathcal A^K(S^n) = \mathcal  A^K(S)\). Intuitively, \(\mathcal A^K(S^k)\) is the region determined by the first \(k\) balls. When adding ball \(B_{k+1}\), each \(\mathcal A^K(S^k_i)\) is either
\begin{enumerate}
    \item fully inside \(B_{k+1}\),
    \item fully outside \(B_{k+1}\), or 
    \item partially intersecting \(B_{k+1}\).
\end{enumerate}

In case (3), we split \(\mathcal A^K(S^k_i)\) into two parts (one inside and one outside), creating two new sets \(S^k_i\) and \(S^k_i \cup \{k+1\}\). We also check whether the single-ball region \(B_{k+1}\) itself (excluding all other balls) is non-empty.

Informally, the intermediate area is the region in which a new candidate could be located such that their presence could affect the preferences of the voters in \( S \), assuming they are only influenced by candidates in the set \( K \). It captures the geometric influence zone of \( K \) over a subset of voters.

The \emph{area-updating method (AUM)} proceeds by applying this iterative splitting.

\subsubsection{Canonical Representation}

To reduce the complexity of AUM, we represent the region \( \mathcal A^K(S) \) using its \emph{canonical} representation, denoted \( A^K(\bd S) \). Here, \(\bd S\) consists of two sets, \(\bd S^+\) and \(\bd S^-\), such that  
\[
\mathcal A^K(S) = \bigcap_{B \in \bd S^+} B \setminus \bigcup_{B \in \bd S^-} B.
\]

Furthermore, \(\bd S^+\) and \(\bd S^-\) can be further divided into hard boundary sets \(\bd S^+_h\) and \(\bd S^-_h\), and soft boundary sets \(\bd S^+ \setminus \bd S^+_h\) and \(\bd S^- \setminus \bd S^-_h\). The hard boundary set consists of balls \(B\) whose sphere \(\varphi\) satisfies the condition:
\begin{align*}
\exists x \in \varphi, \quad &\text{such that every open set } U \text{ containing } x \\
&\text{intersects both } \mathcal{A}^K(S) \text{ and } A^c(S),
\end{align*}
where \(A^c(S)\) is the complement of \(\mathcal A^K(S)\).

Moreover, removing any element \(B\) from \(\bd S\) results in \(A^K(\bd S \setminus \{B\}) \neq \mathcal A^K(S)\), ensuring the minimality of the representation. We can construct \(\mathcal A^K(\bd S)\) from \(\mathcal A^K(S)\) and vice versa. If \(\lvert \bd S\rvert \le k\) for all regions, we can use the canonical representation to describe the area efficiently, reducing computational complexity.

The purpose of the canonical representation is to normalise the candidate set such that the analysis becomes independent of translation, rotation, or scaling. This simplifies the structure of the problem without affecting the outcome of interest (e.g., winner determination or distance calculations).

\subsection{Algorithm for Area Construction}

Algorithm~\ref{iter_all} sketches how to construct all non-empty regions incrementally.

This algorithm identifies the subset of voters for whom a new candidate might feasibly win. Unlike Crama’s method, which examines all \( \binom{n}{d} \)  subsets without filtering, our approach prunes infeasible configurations by leveraging structural properties of the canonical representation of candidates. Crama’s method considers every possible subset of voters indiscriminately, whereas our method takes advantage of the fact that only voters whose preference regions intersect with the convex hull of the candidate set \(K\) can be influenced by the introduction of a new candidate. 

The key idea is to avoid brute-force enumeration of all \( \binom{n}{d} \) subsets by applying a geometric pruning step: we observe that only voters whose preference regions intersect the convex hull of the candidate set \( K \) can be influenced by the introduction of a new candidate. This insight enables early termination in many cases and underpins the practical efficiency of our approach, despite the matching worst-case complexity of \( O(n^{d+1}) \).

\begin{algorithm}[h]
\caption{Area Construction Algorithm}
\label{iter_all}
\begin{algorithmic}[1]
\Procedure{Area\_Construction}{$[B]_n$}
  \Comment{\textit{\small{$[B]$ is a set of $n$ $\ell_2$ balls in $\mathbb{R}^d$.}}}

\State $S\_SET \gets \{\}$
  \Comment{\textit{\small{Stores all canonical sets $(\bd S)$ whose $A^K(\bd S)$ is non-empty.}}}

  \For{$i = 1,\dots,n$}
    \For{$(\bd S) \in S\_SET$}
      \State $Result \gets \textsc{Test\_Intersection\_Relation}((\bd S),\; B_i)$
      \If{$Result = \text{Inside or Outside}$}
         \State \textbf{continue} 
         \Comment{\textit{\small{No change if $A^K(\bd S)$ is fully inside/outside $B_i$.}}}
      \Else
         \Comment{\textit{\small{$B_i$ cuts $A^K(\bd S)$ into two.}}}
         \State $( (\bd S), (\bd S'') ) \gets \textsc{Update\_Boundary}((\bd S), B_i)$
         \State $S\_SET \gets S\_SET \cup \{(\bd S'')\}$ 
         \Comment{\textit{\small{Add new boundary to the set.}}}
      \EndIf
    \EndFor

    \State $(\bd S_{\text{single}}) \gets \textsc{Single\_Boundary}\bigl(B_i,\,[B_1,\dots,B_{i-1}]\bigr)$
    \If{$(\bd S_{\text{single}})\neq \text{NULL}$}
      \State $S\_SET \gets S\_SET \cup \{\bd S_{\text{single}}\}$
    \EndIf
  \EndFor
\EndProcedure
\end{algorithmic}
\end{algorithm}

\vspace{0.5em}
\noindent
\textbf{Subroutines}:
\begin{itemize}
\item \textsc{Test\_Intersection\_Relation}($(\bd S),B$): Determines if $A^K(\bd S)$ is entirely inside $B$, entirely outside $B$, or intersects $B$ non-trivially.
\item \textsc{Update\_Boundary}($(\bd S),B$): Constructs two new boundary sets corresponding to $A^K(\bd S)\cap B$ and $A^K(\bd S)\setminus B$.
\item \textsc{Single\_Boundary}($B_i,[B_1,\dots,B_{i-1}]$): Returns the boundary set describing the region contained only by $B_i$ (and excluded by the others), if non-empty; otherwise returns $\text{NULL}$.
\end{itemize}

\paragraph{Overall Complexity.}
Summarizing the loop structure:
\begin{itemize}
    \item Each ball $B_i$ triggers up to $L$ calls to 
          $\textsc{Test\_Intersection\_Relation}$ and 
          $\textsc{Update\_Boundary}$ (the inner loop).
    \item We also call $\textsc{Single\_Boundary}$ once per ball (the outer loop).
\end{itemize}
Hence, the total subroutine cost is
\[
  O\bigl(n \times L \times \bigl(T_{\text{test}}(k) + T_{\text{update}}(k)\bigr)\bigr)
  \;+\;
  O\bigl(n \times T_{\text{single}}(n)\bigr),
\]
where
\begin{itemize}
  \item $L$ is the number of canonical sets,
  \item $k$ is the maximum canonical size,
  \item $T_{\text{test}}(k)$ and $T_{\text{update}}(k)$ are the running times of $\textsc{Test\_Intersection\_Relation}$ and $\textsc{Update\_Boundary}$, respectively, when the input size is bounded by $k$, and
  \item $T_{\text{single}}(n)$ is the running time of $\textsc{Single\_Boundary}$, which may scale with $n$.
\end{itemize}

\noindent
We defer the precise construction and analysis of these subroutines, including the relationship between $L$ and $k$, to Section~\ref{imMA}.

\subsection{Correctness of the Area-Updating Method (AUM)}

Below, we show that AUM (and thus Algorithm~\ref{iter_all}) indeed finds all non-empty meaningful areas.

\begin{theorem}
    For any finite collection of $n$ balls in $\mathbb{R}^d$, the area-updating method (AUM) finds all non-empty meaningful areas $\mathcal A^K(S)$. 
\end{theorem}

\begin{proof}
\textbf{(Induction on $n$.)}

\noindent
\textbf{Base case} ($n=1$): If there is only one ball $B_1$, AUM obviously finds $\mathcal A^K(\{1\})$ and $\mathcal A^K(\varnothing)$ (the latter is empty unless we count the “outside” region).

\noindent
\textbf{Inductive step}: Suppose AUM works for $n=k$. That is, it generates a family $\fami_k$ of index sets such that, whenever $\mathcal A^K(S^k)$ is non-empty, $S^k$ appears in $\fami_k$. Now add a new ball $B_{k+1}$. Suppose there exists $S^{k+1}$ such that $\mathcal A^K(S^{k+1})$ is non-empty but is not generated by AUM. We consider two cases:

\begin{enumerate}
    \item $k+1 \notin S^{k+1}$. Then $S^{k+1}$ is a subset of $[k]$. If $\mathcal A^K(S^{k+1})$ were non-empty, it should already have appeared at the $n=k$ stage. This contradicts the inductive hypothesis.

    \item $k+1 \in S^{k+1}$. Let $S^k = S^{k+1} \setminus \{k+1\}$. By the inductive hypothesis, $S^k$ is in $\fami_k$ if $\mathcal A^K(S^k)$ is non-empty. Since AUM checks each already-discovered $S^k$ against $B_{k+1}$, if $\mathcal A^K(S^{k+1})$ were non-empty, AUM would split $\mathcal A^K(S^k)$ accordingly, producing $S^{k+1}$. Hence it must be found by AUM, again a contradiction.
\end{enumerate}
Therefore, AUM must generate all $S^{k+1}$ such that $\mathcal A^K(S^{k+1})$ is non-empty, completing the induction.
\end{proof}

Since Algorithm \ref{iter_all} is contingent on the number of represented areas, we prove that for election instances on \(\mathbb{R}^d\) with \(m=1\), this number is bounded by \(O(n^d)\):

The following result is a classical bound in computational geometry; see e.g.\ \cite{Edelsbrunner1987, AgarwalSharir1998}.
In our context, this result bounds the number of distinct voter response regions that must be considered when evaluating the effect of adding a new candidate. This observation underpins the worst-case runtime analysis of our enumeration algorithm in Appendix~\ref{sec:ITR_ALL}.

\begin{theorem}\label{ball_int}
  Let \(n\) be the number of \(d\)-dimensional balls in \(\mathbb{R}^d\).
  Then the number of distinct regions formed by these balls, i.e.\ the number of non-empty sets
  \(
    \mathcal A^K(S) \;=\; \bigcap_{i \in S} B_i \;\;\setminus\;\;\bigcup_{j \in [n]\setminus S} B_j,
  \)
  is at most \(O(n^d)\).
\end{theorem}

\begin{proof}
\textbf{Base Case (\(d=2\)):}  
Consider \(n\) circles in the plane. It is a known fact (and can be shown by a simple planar argument) that each new circle can intersect previously placed circles in at most \(2(k-1)\) points, thus creating at most \(2(k-1)\) new regions when it is added as the \(k\)-th circle. Summing over all \(k\) from \(1\) to \(n\) yields a total of \(O(n^2)\) regions. Thus, for \(d=2\), the bound \(O(n^d) = O(n^2)\) holds.

\smallskip

\textbf{Inductive Step:}  
Suppose for some fixed \(d \ge 2\), every arrangement of \((n-1)\) \(d\)-dimensional balls in \(\mathbb{R}^d\) creates at most \(A^K(d,n-1)\) distinct regions, and that we already know \(\mathcal A^K(d-1, m) = O(m^{d-1})\) for any \(m\). We want to show \(\mathcal A^K(d,n) = O(n^d)\).

Consider adding the \(n\)-th ball \(B_n\) to an existing arrangement of \((n-1)\) balls in \(\mathbb{R}^d\). Let \(\psi_n\) be the \((d-1)\)-dimensional sphere bounding \(B_n\). A new region is created only when \(\psi_n\) \emph{partially} cuts through an existing region \(R\). Equivalently, this requires 
\(
  R \;\cap\; \psi_n 
  \;\neq\; \varnothing.
\)

However, \(R\) itself is defined by the intersection/exclusion of the other \((n-1)\) balls, so the intersection \(R \cap \psi_n\) is determined by an arrangement of \((n-1)\) \((d-1)\)-dimensional spheres (namely \(\{\psi_1 \cap \psi_n,\,\psi_2 \cap \psi_n,\,\dots,\psi_{n-1} \cap \psi_n\}\)) on \(\psi_n\). By the induction hypothesis on dimension, such an arrangement in \((d-1)\)-dimensions can produce at most \(\mathcal A^K(d-1, n-1) = O((n-1)^{d-1})\) connected regions on \(\psi_n\). Hence, ball \(B_n\) can cause at most \(O((n-1)^{d-1})\) new “cuts” in \(\mathbb{R}^d\). Symbolically, if we write \(\mathcal A^K(d,k)\) to denote the maximum number of regions formed by \(k\) \(d\)-dimensional balls, then
\(
   \mathcal A^K(d,n) 
   \;\le\; \mathcal A^K(d,n-1) +\mathcal A^K(d-1,n-1).
\)
Since \(\mathcal A^K(2,n) = O(n^2)\) (the base case), standard iterated arguments (or a straightforward induction on \(d\)) imply \(\mathcal A^K(d,n) = O(n^d)\).  

\smallskip

Because each meaningful area \(\mathcal A^K(S)\) is one of the connected regions formed by the \(n\) balls (differing only in how we label the “inside” and “outside” of each ball), the total number of non-empty \(\mathcal A^K(S)\) is at most \(\mathcal A^K(d,n)\). Therefore, we conclude 
\[
  \text{Number of non-empty } \mathcal A^K(S) \;\le\; \mathcal A^K(d,n) \;=\; O(n^d).
\]
\end{proof}


\subsection{Boundary Finding Algorithm}\label{imMA}

We introduce an algorithm to check whether a new ball \( B_t \), representing a candidate at point \( c \in \mathbb{R}^d \) with radius \( r_t \), partially intersects a region \( \mathcal A^K(S) \). Concretely, we seek two points: one in \( B_t \cap \mathcal A^K(S) \), and another in \( \mathcal A^K(S) \setminus B_t \). This test is essential for detecting when a candidate may lie near the boundary of a feasible winning region.

Our key contribution is a boundary-checking procedure based on solving two constrained quadratic programs (one minimisation, one maximisation), using geometric properties of metric balls. These programs determine whether the candidate ball intersects the feasible region non-trivially.

The algorithm improves on naive geometric enumeration by pruning infeasible configurations through convexity arguments. Although related techniques appear in computational geometry, our setting is adapted to preference structures induced by overlapping metric balls, requiring different filtering logic.

This can be reformulated as two optimisation problems. First, we minimise
\[
\sum_{j=1}^d (x_j - c_{t,j})^2
\]
subject to \(x \in B_i\) for all \(B_i \in \bd^+ S\) and \(x \notin B_j\) for all \(B_j \in \bd^- S\). If the optimal value \(m\) is less than \(r_t^2\), then \(B_t \cap \mathcal A^K(\bd S)\neq\emptyset\). Second, we maximise
\[
\sum_{j=1}^d (x_j - c_{t,j})^2
\]
under the same constraints. If the optimal value \(M\) exceeds \(r_t^2\), then \(\mathcal A^K(\bd S)\setminus B_t\neq\emptyset\). However, neither \(B_t \cap \mathcal A^K(\bd S)\) nor \(\mathcal A^K(\bd S) \setminus B_t\) need be convex, so one cannot always rely on efficient polynomial-time methods; additional modifications or heuristics may be necessary for tractability.

To implement the optimisation-based test efficiently, we define the relevant geometric constructs:

\begin{definition}
\ \\
\vspace{-1.2em} 

\begin{itemize}
    \item \( \varphi_i \): The boundary (sphere) of a \( d \)-dimensional ball \( B_i \).
    \item \( \varphi_{i,j} \): The intersection of \(\varphi_i\) and \(\varphi_j\), which forms a \((d-1)\)-dimensional sphere.
    \item \( H_{i,j} \): The hyperplane containing \(\varphi_{i,j}\). The half-space cut by \( H_{i,j} \) that contains \(\varphi_i \cap B_j\) is denoted \(H^+_{i,j}\), while its complement (including \(H_{i,j}\) itself) is denoted \(H^-_{i,j}\). 
        \begin{itemize}
            \item If \(\varphi_i \cap \varphi_j = \emptyset\), define \(H_{i,j} = \emptyset\).
            \item If \(\varphi_i \subset B_j\), define \(H^+_{i,j} = \mathbb{R}^d\) and \(H^-_{i,j} = \emptyset\).
            \item Otherwise, define \(H^+_{i,j} = \emptyset\) and \(H^-_{i,j} = \mathbb{R}^d\).
        \end{itemize}
    \item \( H^+(y) \): Given a half-space \(H^+\) defined by 
    \(\sum_{i=0}^d h_i x_i + c \;\;\ge 0,\)
    we define \( H^+(y) \;=\; \sum_{i=0}^d h_i\,y_i + c\).
\end{itemize}
\end{definition}

We now present Algorithm~\ref{Meyers}, which tests each input ball as a potential boundary via a geometric filtering step. The algorithm relies on a subroutine, \textsc{Single\_Boundary\_Detection}, to check whether a given ball splits the feasible region.

\begin{algorithm}[h]
   \caption{Finding Boundary of Common Intersection Area}
   \label{Meyers}
   \begin{algorithmic}[1]
     \Procedure{Boundary\_Finding}{$[B]^{in}_n$, $[B]^{out}_m$}
       \Comment{\textit{\small{$[B]^{in/out}$ are sets of \(\ell_2\)-balls in \(\mathbb{R}^d\). 
       Goal: check if 
       \(\bigcap [B]^{in} \,\setminus\, \bigcup [B]^{out} \;=\;\emptyset\).}}}
       
       \For{each ball \(B_i\) in \([B]^{in}_n \;\cup\; [B]^{out}_m\)}
         \Comment{\textit{\small{Test each ball as a potential boundary.}}}
         \If{\textsc{Single\_Boundary\_Detection}($B_i,\,[B]^{in}_n \setminus\{B_i\},\,[B]^{out}_m \setminus\{B_i\},\,d$) = True}
            \State \Return \textbf{true}
            \Comment{\textit{\small{Finding one suitable boundary suffices.}}}
         \EndIf
       \EndFor
       
       \State \Return \textbf{false}
       \Comment{\textit{\small{No valid boundary was found.}}}
     \EndProcedure
     \vspace{1em}

     \Procedure{Single\_Boundary\_Detection}{$B_t,\;[B]^{in}_n,\;[B]^{out}_m,\;d$}
       \State \(\displaystyle [H^+] \gets \bigl\{H^+_{t,i} : B_i \in [B]^{in}_n \bigr\}\)
         \Comment{\textit{\small Compute half-spaces \(H^+_{t,i}\).}}
       \State \(\displaystyle [H^-] \gets \bigl\{H^-_{t,i} : B_i \in [B]^{out}_m \bigr\}\)
         \Comment{\textit{\small For out-balls, take complementary half-spaces.}}
       
       \State \(\displaystyle \textit{object} \gets \sum_{j=1}^d (x_j - c_{t,j})^2\)
         \Comment{\textit{\small Quadratic form measuring distance from the centre \(c_t\).}}
       
       \State \( m = \textsc{Optimiser}(\mathrm{min},\;\textit{object},\;[H^+ \cup H^-])\)
       \State \( M = \textsc{Optimiser}(\mathrm{max},\;\textit{object},\;[H^+ \cup H^-])\)
         \Comment{\textit{\small Min and max distance-squared.}}
       
       \If{ $m$ and $M$ both exist \textbf{and} $m < r_t^2 < M$ }
         \State \Return \textbf{true}
         \Comment{\textit{\small The feasible set intersects the sphere of radius \(r_t\).}}
       \Else
         \State \Return \textbf{false}
       \EndIf
     \EndProcedure
   \end{algorithmic}
\end{algorithm}

\noindent
If \( \mathsf{QP_{min}} \) has a positive-definite Hessian, it can be solved efficiently using interior point methods; see~\cite{goldfarb1991primal, leulmi2024interior}. The best known complexity is \( O(d^\omega n L) \), where \( \omega \approx 2.373 \). For constant \( d \), this becomes linear in \( n \). 

For \( \mathsf{QP_{max}} \), the objective is not convex and the problem is NP-hard in general. However, by bounding the number of active constraints \( k \) in the feasible region, we achieve a runtime of \( O(n^2 k^d) \). This combinatorial scaling makes our method tractable in settings where the feasible region has low complexity.

In summary, Algorithm~\ref{Meyers} implements an efficient boundary-checking mechanism combining convex geometry, quadratic optimisation, and structure-aware pruning. This enables scalable testing of voter subsets in high dimensions and contributes a novel geometric filtering mechanism for preference-based decision models.

\section{Algorithm~\ref{MQP}: Feasibility Check for Quadratic Optimisation}

\begin{algorithm}[h]
\caption{Feasibility Check for Quadratic Optimisation}
\label{MQP}
\begin{algorithmic}[1]
\Require Goal function \( X^T Q X + P^T X \), threshold \( c \), constraints \( AX + C > 0 \) in \( d\)-dimensional space
\Ensure \textbf{true} if there exists an \(X\) such that \(\bigl(X^T Q X + P^T X\bigr) > c\) and \(AX + C > 0\); otherwise \textbf{false}

\Procedure{CheckFeasibility}{}
    \For{each \(i=1,\ldots,d\)}
       \State Solve:
       \[
         \max x_i \quad \text{subject to } A X + C > 0,
         \quad
         \min x_i \quad \text{subject to } A X + C > 0.
       \]
       \Comment{\textit{\small Each LP runs in \(O(n\,d^\omega)\), total \(2d\) runs \(\implies O(n\,d^{\omega+1})\).}}
       \If{any solution is unbounded}
          \State \Return \textbf{true} \Comment{\textit{\small Quadratic can be unbounded in that direction.}} 
       \EndIf
    \EndFor
    
    \State Let \(X_0\) be a feasible vertex found in the LP phase.
    \If{ \(X_0^T Q X_0 + P^T X_0 > c\)}
       \State \Return \textbf{true}
       \Comment{\textit{\small Already above threshold.}}
    \EndIf
    
    \State \(V \gets \{X_0\}, \quad \mathcal{Q} \gets \{X_0\}\)
      \Comment{\textit{\small Traverse vertices of the feasible region.}}
    
    \While{\(\mathcal{Q}\) is not empty}
        \State \( X_{\mathrm{current}} \gets \text{dequeue from }\mathcal{Q}\)
        \State \(\text{AdjVertices} \gets \textsc{FindAdjacentVertices}\bigl(X_{\mathrm{current}},\,A,\,C\bigr)\)
        
        \For{each \( X_{\mathrm{next}} \in \text{AdjVertices}\)}
            \If{ \(X_{\mathrm{next}} \notin V\)}
                \If{ \(X_{\mathrm{next}}^T Q X_{\mathrm{next}} + P^T X_{\mathrm{next}} > c\)}
                    \State \Return \textbf{true}
                \EndIf
                \State \(V \gets V \cup \{X_{\mathrm{next}}\}\)
                \State \(\mathcal{Q}\) \text{enqueue } \(X_{\mathrm{next}}\)
            \EndIf
        \EndFor
    \EndWhile
    
    \Return \textbf{false}
\EndProcedure

\vspace{1em}

\Procedure{FindAdjacentVertices}{$X_0,\;A,\;C$}
    \State \(\displaystyle \mathcal{A} \;\gets\; \{\,i \,\mid\, A_i X_0 + C_i = 0\,\}\)
    \Comment{\textit{\small Indices of active constraints at \(X_0\).}}
    \State \(\text{AdjVertices} \gets \emptyset\)
    
    \For{each subset \(S \subseteq \mathcal{A}\) of size \(d-1\)}
        \State Solve \(A_S\,v = 0\) for \(v\) 
        \Comment{\textit{\small Solving \(d-1\) linear equations: \(O(d^3)\).}}
        
        \For{each inactive constraint \(A_j X + C_j > 0\)}
            \State Solve \(A_j\,(X_0 + \lambda\,v) + C_j = 0\) for \(\lambda\)
            \Comment{\textit{\small One-dimensional root finding in \(O(d)\).}}
            
            \If{\(\lambda > 0\)}
                \State \(X_{\mathrm{next}} \gets X_0 + \lambda\,v\)
                \If{\(X_{\mathrm{next}}\) satisfies all \(A X + C > 0\)}
                    \State \(\text{AdjVertices} \gets \text{AdjVertices}\,\cup\,\{X_{\mathrm{next}}\}\)
                \EndIf
            \EndIf
        \EndFor
    \EndFor
    
    \Return \(\text{AdjVertices}\)
\EndProcedure
\end{algorithmic}
\end{algorithm}

\paragraph{Complexity Analysis.}

\begin{itemize}
\item \textbf{Step 1:} We solve \(2d\) linear programmes, each in \(O(n\,d^\omega)\) time, for a total \(O(n\,d^{\omega+1})\).  
\item \textbf{Step 2:} We do a graph-like traversal of the vertices. If \(K\) constraints (out of \(n\)) define the boundary, there are at most \(O(K^d)\) potential vertices. Checking adjacency costs up to \(O(n d^4\,k^{d})\) in the worst case. Evaluating \(X^T Q X + P^T X\) at each vertex adds \(O(d^2)\). 
Overall:
\[
   O\bigl(n\,d^{\omega+1}\bigr) \;+\; O\bigl(n d^4\,K^{d}\bigr).
\]
For large \(d\), the exponential term can dominate; the procedure is mainly practical for moderate \(d\) or \(K\).

\item If \(K<d\), the feasible region is unbounded in at least one direction. Hence \(\mathsf{QP_{max}}\) is effectively unbounded if feasible.

\item When there is only one candidate ($m=1$), each hyperplane intersects at one point (The position of the candidate). The feasible area, if it exists, will become cone-shaped. As a result, \(\mathsf{QP_{max}}\) would also be unbounded.

\end{itemize}

\subsection*{Implementation of Subroutines for Area Construction Algorithm}

Below we outline how to use the previously discussed procedures:

\begin{itemize}
\item \textsc{Test\_Intersection\_Relation}($(\bd S),B$):  
  \begin{enumerate}
  \item Run \(\textsc{Boundary\_Finding}\bigl(\bd^+ S \cup \{B\},\,\bd^-S\bigr)\). 
    \begin{itemize}
    \item If it returns \(\textbf{false}\), then \(\mathcal A^K(\bd S)\cap B=\emptyset\) (i.e.\ \textsc{Outside}).
    \end{itemize}
  \item Otherwise, run \(\textsc{Boundary\_Finding}\bigl(\bd^+ S,\,\bd^-S \cup \{B\}\bigr)\).
    \begin{itemize}
    \item If it returns \(\textbf{false}\), then \(\mathcal A^K(\bd S)\setminus B=\emptyset\) (i.e.\ \textsc{Inside}).
    \end{itemize}
  \item Otherwise, return \textsc{Intersect}.
  \end{enumerate}

\item \textsc{Update\_Boundary}($(\bd S),B$):
  \begin{enumerate}
    \item For each \(B_i \in \bd S\), call
    \[
      \textsc{Test\_Intersection\_Relation}\bigl((\bd S)\setminus\{B_i\}\cup\{B\},\,B_i\bigr)
    \]
    to test whether \(B_i\) remains a boundary for \(\mathcal A^K(\bd S)\cap B\).
    \item Similarly, test if \(B_i\) remains a boundary for \(\mathcal A^K(\bd S)\setminus B\).
    \item Collect these updates to form the boundary sets for the ``inside'' and ``outside'' parts.
  \end{enumerate}

\item \textsc{Single\_Boundary}($B_i,[B_1,\dots,B_{i-1}]$):
  \begin{enumerate}
    \item Run \(\textsc{Boundary\_Finding}\bigl(\{B_i\},\,\{B_1,\dots,B_{i-1}\}\bigr)\) to check if \(\mathcal A^K(\{B_i\})\) is non-empty.
    \item If non-empty, test each \(B_\ell \in \{B_1,\dots,B_{i-1}\}\) (via \textsc{Test\_Intersection\_Relation}) to see if it belongs to the boundary of \(\mathcal A^K(B_i)\).
  \end{enumerate}
\end{itemize}

\subsubsection*{Complexity Analysis of the Subroutines}
Below we assume:
\begin{itemize}
  \item \(\textsc{Boundary\_Finding}\) itself costs \(O(n^2\,K^d)\).
  \item \(K\) (in the canonical set \(\bd S\)) is at most some fixed bound, and
  \item \(n\) is the total number of balls.
\end{itemize}

\paragraph*{\textsc{Test\_Intersection\_Relation}($(\bd S),B$).}
The procedure calls \(\textsc{Boundary\_Finding}\) twice (once for ``$\cap B$'' and once for ``$\setminus B$''). Hence,
\[
  \textsc{Test\_Intersection\_Relation} = O\bigl(n^2\,K^d\bigr).
\]

\paragraph*{\textsc{Update\_Boundary}($(\bd S),B$).}
We iterate over each \(B_i\in (\bd S)\), where \(\lvert \bd S\rvert \le k\). For each ball \(B_i\), we do a call to \textsc{Test\_Intersection\_Relation} (or a small constant number of calls) to decide whether \(B_i\) remains in the boundary. Each such call is \(O(n^2\,K^d)\), so in total:
\[
  \textsc{Update\_Boundary}\bigl((\bd S),B\bigr)
  \;=\;
  O\bigl(k\,n^2\,K^d\bigr).
\]

\paragraph*{\textsc{Single\_Boundary}($B_i,[B_1,\dots,B_{i-1}]$).}
\begin{enumerate}
\item We first call \(\textsc{Boundary\_Finding}(\{B_i\},\,\{B_1,\dots,B_{i-1}\})\), costing \(O(n^2\,K^d)\).
\item If the region is non-empty, we then check each of the previous balls (up to \(n\) of them) with \(\textsc{Test\_Intersection\_Relation}\), each costing \(O(n^2\,K^d)\).

Hence the worst-case is
\[
   O\bigl(n^2\,K^d\bigr)
   \;+\;
   O\bigl(n \times n^2\,K^d\bigr)
   \;=\;
   O\bigl(n^3\,K^d\bigr).
\]
\end{enumerate}

\end{document}